\documentclass[a4paper,USenglish,cleveref, autoref, thm-restate]{lipics-v2021}

\pdfoutput=1 
\hideLIPIcs  

\title{On the Expressive Power of Regular Expressions with Backreferences}

\author{Taisei Nogami}{Waseda University, Tokyo, Japan}{sora410@fuji.waseda.jp}{}{}

\author{Tachio Terauchi}{Waseda University, Tokyo, Japan \and \url{https://www.f.waseda.jp/terauchi/}}{terauchi@waseda.jp}{https://orcid.org/0000-0001-5305-4916}{}

\authorrunning{T. Nogami and T. Terauchi}

\Copyright{Taisei Nogami and Tachio Terauchi}



\category{} 

\nolinenumbers 

\funding{This work was supported by JSPS KAKENHI Grant Numbers JP20H04162, JP20K20625, and JP22H03570.}

\EventEditors{J\'{e}r\^{o}me Leroux, Sylvain Lombardy, and David Peleg}
\EventNoEds{3}
\EventLongTitle{48th International Symposium on Mathematical Foundations of Computer Science (MFCS 2023)}
\EventShortTitle{MFCS 2023}
\EventAcronym{MFCS}
\EventYear{2023}
\EventDate{August 28 to September 1, 2023}
\EventLocation{Bordeaux, France}
\EventLogo{}
\SeriesVolume{272}
\ArticleNo{71}

\usepackage{amsmath}
\usepackage{amssymb}
\usepackage{amsthm}
\usepackage[normalem]{ulem}
\usepackage{comment}
\usepackage{xcolor}
\definecolor{mydeepred}{HTML}{c51b1d}
\usepackage{tikz-cd}
\theoremstyle{definition}
\newtheorem{definition2}[theorem]{Definition}

\newcommand{\bs}{\backslash}

\newcommand{\suuretu}[1]{\ensuremath{\left\{#1\right\}}}%
\makeatletter
\def\syuugou{\@ifnextchar[{\@syuugou}{\suuretu}}
\def\@syuugou[#1]#2{\ensuremath{\left\{\left.#1\,
  \@ifundefined{hakobanpush}{}{\hakobanpush}%
  \vphantom{#2}\@ifundefined{hakobanpush}{}{\hakobanpop}%
  \right |\,%
  #2\right\}}}
\makeatother

\newcommand{\zettaiti}[1]{\left|#1\right|}
\renewcommand{\uuline}[1]{\underline{\underline{#1}}}

\newcommand{\slice}[1]{[#1]}
\DeclareMathOperator{\varf}{var}
\newcommand{\rewb}{\textsl{REWB}}
\newcommand{\refwords}{\Sigma ^ {[\ast]}}
\DeclareMathOperator{\cnt}{cnt}
\newcommand{\powerset}[1]{\mathcal{P}(#1)}

\newcommand{\mydollar}{\textup{\textdollaroldstyle}}
\newcommand{\mycent}{\textup{\textcentoldstyle}}
\newcommand{\myleft}{\mathrm{L}}
\newcommand{\myright}{\mathrm{R}}
\newcommand{\mystop}{\mathrm{S}}

\newcommand{\myid}[1]{\vdash_{#1}}
\newcommand{\myspl}{\!\upharpoonleft}
\newcommand{\myspr}{\upharpoonright}
\newcommand{\mydtoz}[1]{\overline{#1}}

\newcommand{\myonlymove}[1]{\models_{#1}}
\newcommand{\mynat}{\mathbb{N}}
\newcommand{\mydisjointu}{\uplus}

\newcommand{\symcall}{c}
\newcommand{\symexec}{e}
\newcommand{\symret}{r}

\newcommand{\reflang}{\mathcal{R}}
\newcommand{\deref}{\mathcal{D}}

\newcommand{\dir}{\Delta}

\newcommand{\maru}[1]{\raise0.16ex\hbox{\textcircled{\scriptsize{#1}}}}
\newcommand{\mylipicsenumitem}[1]{\text{\textcolor{lipicsGray}{\sffamily\bfseries\upshape{(#1)}}}}

\newcommand{\myto}[2]{\underset{#2}{\overset{#1}{\longrightarrow}}}
\newcommand{\mytto}[2]{\underset{#2}{\overset{#1}{\Longrightarrow}}}

\newcommand{\figto}[2]{{#1}\!\to\!{#2}}
\usepackage{tikz}
\usetikzlibrary{arrows.meta, automata, positioning,decorations.pathmorphing,shapes}

\newcommand{\stkout}[1]{\ifmmode\text{\sout{\ensuremath{#1}}}\else\sout{#1}\fi}

\newcommand{\fullversion}[2]{{#2}}

\newcommand{\phantomcomment}[1]{}

\makeatletter
\newcommand\niton{\mathrel{\m@th\mathpalette\canc@l\owns}}
\newcommand\canc@l[2]{{\ooalign{$\hfil#1/\mkern1mu\hfil$\crcr$#1#2$}}}
\makeatother

\allowdisplaybreaks
 \renewcommand{\fullversion}[2]{{#1}}
\begin{document}

\maketitle

\begin{abstract}
    A \emph{rewb} is a regular expression extended with a feature called \emph{backreference}. It is broadly known that backreference is a practical extension of regular expressions, and is supported by most modern regular expression engines, such as those in the standard libraries of Java, Python, and more. Meanwhile, \emph{indexed languages} are the languages generated by indexed grammars, a formal grammar class proposed by A.V.Aho. We show that these two models' expressive powers are related in the following way: every language described by a rewb is an indexed language. As the smallest formal grammar class previously known to contain rewbs is the class of context sensitive languages, our result strictly improves the known upper-bound. Moreover, we prove the following two claims: there exists a rewb whose language does not belong to the class of stack languages, which is a proper subclass of indexed languages, and the language described by a rewb without a captured reference is in the class of nonerasing stack languages, which is a proper subclass of stack languages. Finally, we show that the hierarchy investigated in a prior study, which separates the expressive power of rewbs by the notion of nested levels, is within the class of nonerasing stack languages.
\end{abstract}
 
\section{Introduction}
\label{sec:intro}

A \emph{rewb} is a regular expression empowered with a certain extension, called \emph{backreference}, that allows preceding substrings to be used later. It is closer to practical regular expressions than the pure ones, and supported by the standard libraries of most modern programming languages.
 A typical example of a rewb follows:
	\begin{example}
		Let $\Sigma$ be the alphabet $\syuugou{a,b}$. The language $L(\alpha)$ described by the rewb $\alpha = (_1 (a + b) ^ \ast)_1\,\bs 1$ is $\syuugou[ww]{w \in \Sigma^\ast}$.
        Intuitively, $\alpha$ first \emph{captures} a preceding string $w \in L((a+b)^\ast)$ by $(_1\,)_1$, and second \emph{references} that $w$ by following $\bs 1$. Therefore, $\alpha$ matches $ww$. Because this $L(\alpha)$ is a textbook example of a non-context-free language (and therefore non-regular), the expressive power of rewbs exceeds that of the pure ones.
	\end{example}
	
    In 1968, A.V.Aho discovered indexed languages with characterizations by two equivalent models: indexed grammars and (one-way\footnote{``One-way'' means that the input cursor will not move back to left. The antonym is ``two-way.''} nondeterministic, or 1N)\,nested stack automata\,(NSA)~\cite{aho1968indexed, aho1969nested}. The class of indexed languages is a proper superclass of context free languages\,(CFL), and a proper subclass of context sensitive languages\,(CSL)~\cite{aho1968indexed}.
	
	Berglund and van der Merwe~\cite{berglund2023re}, and C{\^a}mpeanu et al.~\cite{campeanu2003formal} have shown that the class of rewbs is incomparable with the class of CFLs and is a proper subclass of CSLs. As the first main contribution of this paper, we prove that the language described by a rewb is an indexed language. Since the class of CSLs was the previously known best upper-bound of rewbs, our result gives a novel and strictly tighter upper-bound.
		
    Meanwhile, there is a class of the languages called stack languages~\cite{ginsburg1967stack, ginsburg1967one}. This class corresponds to the model (1N)\,stack automata\,(SA), a restriction of NSA. Hence, it trivially follows that the class of stack languages is a subclass of indexed languages. Actually, this containment is known to be proper~\cite{aho1969nested}.
	Furthermore, a model called nonerasing stack automata\,(NESA) has been studied in papers such as \cite{ginsburg1967stack,hopcroft1967nonerasing,ogden1969intercalation}, and its language class is known to be a proper subclass of stack languages~\cite{ogden1969intercalation}.

    In this paper, we show that every rewb without a captured reference (that is, one in which no reference $\bs i$ appears as a subexpression of an expression of the form $(_j \alpha )_j$) describes a nonerasing stack language. Given our result, the following question is natural: does every rewb describe a (nonerasing) stack language? We show that the answer is no. Namely, we show a rewb that describes a non-stack language.
    Finally, Larsen~\cite{larsen1998regular} has proposed a notion called \emph{nested levels} of a rewb and showed that they give rise to a concrete increasing hierarchy of expressive powers of rewbs by exhibiting, for each nested level $i \in \mynat$, a language $L_i$ that is expressible by a rewb at level $i$ but not at any levels below $i$. We show that this hierarchy is within the class of nonerasing stack languages, that is, there exists an NESA $A_i$ recognizing $L_i$ for every nested level $i$.
Below, we summarize the main contributions of the paper.
	\begin{alphaenumerate}
		\item Every rewb describes an indexed language.\,(Section~\ref{sec:rewbisig}, Corollary~\ref{cor:rewbisig})
        \item Every rewb without a captured reference describes a nonerasing stack language.\\
		(Section~\ref{sec:rewbisig}, Corollary~\ref{cor:refcap})
        \item There exists a rewb that describes a non-stack language.\,(Section~\ref{sec:rewbisnotsl}, Theorem~\ref{thm:antisl})
		\item The hierarchy given by Larsen~\cite{larsen1998regular} is within the class of nonerasing stack languages.\\(Section~\ref{sec:larsen98isnonerasingsl}, Theorem~\ref{thm:hier})
	\end{alphaenumerate}
Note that by (b) and (c), it follows that there is a rewb that \emph{needs} capturing of references (Section~\ref{sec:rewbisnotsl}, Corollary~\ref{cor:needrefcap}). See also Figure~\ref{fig:langrels} for a summary of the results.

The rest of the paper is organized as follows. Section~\ref{sec:related} discusses related work. Section~\ref{sec:prelim} defines preliminary notions used in the paper such as the syntax and semantics of rewb, SA, NESA, and NSA. Sections~\ref{sec:rewbisig}, \ref{sec:rewbisnotsl}, and \ref{sec:larsen98isnonerasingsl} formally state and prove the paper's main contributions listed above. Section~\ref{sec:conc} concludes the paper with a discussion on future work. For space, the proofs are in \fullversion{Appendix}{the full paper~\cite{nogami2023expressive}}.
 \section{Related Work}
\label{sec:related}

First, we discuss related work on rewbs. There are several variants of the syntax and semantics of rewbs since they first appeared in the seminal work by Aho~\cite{10.5555/114872.114877}. A recent study by Berglund and van der Merwe~\cite{berglund2023re} summarizes the variants and the relations between them. In sum, there are two variants of the syntax, whether or not a same label may appear as the index of more than one capture (``may repeat labels'', ``no label repetitions''), and two variants of the semantics, whether an unbound reference is interpreted as the empty string or an undefined factor ($\varepsilon\text{-semantics}, \emptyset\text{-semantics}$). As shown in \cite{berglund2023re}, there is no difference in the expressive powers between these two semantics under the ``may repeat labels'' syntax (therefore, there are three classes with different expressive powers, namely ``no label repetitions'' with $\emptyset\text{-semantics}$, ``no label repetitions'' with $\varepsilon\text{-semantics}$, and ``may repeat labels''). In this paper, we focus on the ``may repeat labels'' formalization, which has the highest expressive power of the three and is often studied in formal language theory. We adopt the $\varepsilon\text{-semantics}$ as the semantics of rewbs. Note that the pioneering formalization of rewbs given by Aho~\cite{10.5555/114872.114877} has the equivalent expressive power as this class. The rewbs with ``may repeat labels'' with $\varepsilon\text{-semantics}$ was recently proposed by Schmid with the notion of ref-words and dereferences~\cite{schmid2016characterising}. Simultaneously, he proposed a class of automata called \emph{memory automata}\,(MFA), and showed that its expressive power is equivalent to that of rewbs. Freydenberger and Schmid extended MFA to \emph{MFA with trap-state}~\cite{freydenberger2019deterministic}.
Berglund and van der Merwe~\cite{berglund2023re} showed that the class of Schmid's rewbs is a proper subclass of CSLs, and is incomparable with the class of CFLs. Note that there is a pumping lemma for the formalization given by C\^{a}mpeanu et al.~\cite{campeanu2003formal} but it is known not to work for Schmid's rewbs. As mentioned above, Larsen introduced the notion of nested levels and showed that increase in the levels increases the expressive powers of rewbs~\cite{larsen1998regular}.

Next, we discuss related work on the three automata used throughout the paper, namely SA, NESA, and NSA. Ginsburg et al. introduced SA as a mathematical model that is more powerful than pushdown automaton\,(PDA), and NESA as a restricted version of SA~\cite{ginsburg1967stack}. Hopcroft and Ullman discovered a type of Turing machine corresponding to the class of two-way NESA~\cite{hopcroft1967nonerasing}. Ogden proposed a pumping lemma for stack languages and nonerasing stack languages~\cite{ogden1969intercalation}. Aho proposed NSA with a proof of the fact that (1N)\,NSA and indexed grammars given by himself in \cite{aho1968indexed} are equivalent in their expressive powers, and recognized PDA and SA as special cases of NSA~\cite{aho1969nested}. Aho also showed that the class of indexed languages is a proper superclass of CFLs, and a proper subclass of CSLs~\cite{aho1968indexed}. Hayashi proposed a pumping lemma for indexed languages~\cite{hayashi1973derivation}.
 \section{Preliminaries}
\label{sec:prelim}

In this section, we formalize the syntax and the semantics of rewbs following the formalization given in \cite{freydenberger2019deterministic}.  We begin with the syntax. Let $\Sigma_\varepsilon = \Sigma \mydisjointu \syuugou{\varepsilon}$ and $\slice{k} = \syuugou{1,2,\dots, k}$, where the symbol $\mydisjointu$ denotes a disjoint union.
    \begin{definition2} For each natural number $k \geq 1$, the set of \emph{$k$-rewbs} over $\Sigma$, written $\rewb_k$, and the mapping $\varf:\, \rewb_k \rightarrow \mathcal{P}({\slice{k}})$ are defined as follows, where $a \in \Sigma _ {\varepsilon}$ and $i \in \slice{k}$:
	\begin{align*}
		(\alpha, \varf(\alpha)) ::= &(a, \emptyset) \mid (\bs i, \syuugou{i}) \mid (\alpha_0 \alpha_1, \varf(\alpha_0) \cup \varf(\alpha_1)) \mid (\alpha_0 + \alpha_1, \varf(\alpha_0) \cup \varf(\alpha_1)) \\
		& \mid (\alpha_0^\ast, \varf(\alpha_0)) \mid ((_{j} \alpha_0 )_{j}, \varf(\alpha_0) \mydisjointu \syuugou{j})\ \text{where}\ j \in \slice{k} \backslash \varf(\alpha_0).
	\end{align*}
	We also write $\rewb_0$ for the set REG of regular expressions over $\Sigma$, and $\rewb$ for the set of all rewbs, namely $\bigcup _ {k \geq 0} \rewb_k$.
	\end{definition2}
	
	\begin{example}
	For example, $\varepsilon$, $a$, $\bs 1$, $a^\ast\bs 1$, $(_1 a^\ast)_1$, $((_1 a^\ast)_1) ^ \ast$, $(_2 a^\ast)_2\bs 2$, $(_1 a^\ast )_1 (_2 b^\ast)_2 (\bs 1+\bs 2)$, $(_2 (_1 (a+b)^\ast )_1\bs 1 )_2\,\bs 2\,(_2 \bs 1{)_2} ^\ast$, $((_1 \bs 4\,a )_1\,(_2 \bs 3 )_2\,(_3 \bs 2\,a)_3\,(_4 \bs 1 \bs 3 )_4)^\ast$ are rewbs. On the other hand, $(_1 (_1 a^\ast )_1 )_1$, $(_1 a^\ast\,\bs 1 ) _1$, $(_1 (_2 (_1 a^\ast )_1 )_2 )_1$ are not rewbs.
	\end{example}
	
	Note that this syntax allows multiple occurrences of captures with the same label, that is, we adopt the ``may repeat labels'' convention. Next, we define the semantics.
	\begin{definition2} \label{dfn:reflang} 
        Let $B_k = \syuugou[\lbrack_i, \rbrack_i]{i \in [k]}$. The mapping $\reflang_k:\, \rewb_k \rightarrow \mathcal{P}((\Sigma \mydisjointu B_k \mydisjointu [k]) ^ \ast)$ is defined as follows, where $a \in \Sigma _ {\varepsilon}$ and $i \in \slice{k}$:
	\begin{align*}
		&\reflang_k(a) = \syuugou{a},\ \reflang_k(\bs i) = \syuugou{i},\ \reflang_k(\alpha_0 \alpha_1) = \reflang_k(\alpha_0)\reflang_k(\alpha_1),\\
        &\reflang_k(\alpha_0 + \alpha_1) = \reflang_k(\alpha_0) \cup \reflang_k(\alpha_1),\ \reflang_k(\alpha^\ast) = \reflang_k(\alpha)^\ast,\ \reflang_k((_i \alpha )_i) = \syuugou{\lbrack_i}\reflang_k(\alpha)\syuugou{\rbrack_i}.
	\end{align*}
        We let $\refwords_k$ denote $\bigcup _ {\alpha \in \rewb_k} \reflang_k{(\alpha)}$.
	\end{definition2}
	
	\begin{example} $\reflang_k((_1 (a+b)^\ast)_1\bs 1) = \syuugou{\lbrack_1}\syuugou{a,b}^\ast\syuugou{\rbrack_1}\syuugou{1} = \syuugou[\lbrack_1\,w\,\rbrack_1\,1]{w \in \syuugou{a,b} ^ \ast}.$
	\end{example}

    That is, we first regard a rewb $\alpha$ over $\Sigma$ as a regular expression over $\Sigma \mydisjointu B_k \mydisjointu \slice{k}$, deducing the language $\reflang_k(\alpha)$. The second step, described next, is to apply the \emph{dereferencing (partial) function} $\deref_k:(\Sigma \mydisjointu B_k \mydisjointu \slice{k}) ^ \ast \rightharpoonup \Sigma ^ \ast$ to each of its element.

    We give an intuitive description of  $\deref_k$. First, $\deref_k$ scans its input string from the beginning toward the end, seeking $i \in \slice{k}$. If such $i$ is found, $\deref_k$ replaces this $i$ with the substring obtained by removing the brackets in $v$ that comes from the preceding $\lbrack _ i\,v\,\rbrack_i$ if $\lbrack_i$ exists (if this $\lbrack_i$ has no corresponding $\rbrack_i$, $\deref_k$ becomes undefined). Otherwise, $\deref_k$ replaces this $i$ with $\varepsilon$. The dereferencing function $\deref_k$ repeats this procedure until all elements of $\slice{k}$ appearing in the string are exhausted, then removes all remaining brackets. We let $v_{[r]}$ denote the string which $\deref_k$ scans at the $r$\textsuperscript{th} number $n_r \in \slice{k}$ at the $r$\textsuperscript{th} loop (see \fullversion{Appendix~\ref{app:prelim}}{the full version~\cite{nogami2023expressive}} for the formal definitions of $\deref_k$ and $v_{[r]}$).
		\begin{enumerate}
            \item $\lbrack_1 a\,\lbrack_2 b \rbrack_2\, 2\,\rbrack_1\,1$. In this example, $\deref_k$ encounters $n_1=2$ first, and this $2$ corresponds the preceding $\lbrack_2 b \rbrack_2$, therefore this $2$ is replaced with $v_{[1]}=b$. As a result, the input string becomes $\lbrack_1 a\,\lbrack_2 b \rbrack_2\, b\,\rbrack_1\,1$. $\deref_k$ repeats this process again. Now, $\deref_k$ locates $n_2=1$ corresponding the preceding $\lbrack_1 a\,\lbrack_2 b \rbrack_2\, b\,\rbrack_1$, so this $1$ is replaced with $v_{[2]}=a \lbrack_2 b \rbrack_2 b$ but with the brackets erased. Therefore we gain $\lbrack_1 a\,\lbrack_2 b \rbrack_2\, b\,\rbrack_1\,abb$. Finally, $\deref_k$ removes all remaining brackets and produces $abbabb$. Here is the diagram: $\lbrack_1 a\,\underline{\lbrack_2 b \rbrack_2}\, 2\,\rbrack_1\,1 \to \underline{\lbrack_1 a\,\lbrack_2 b \rbrack_2\, b\,\rbrack_1}\,1 \to \lbrack_1 a\,\lbrack_2 b \rbrack_2\, b\,\rbrack_1\,abb \to abbabb.$
			\item $\lbrack_1 a \rbrack_1\,1\,\lbrack_1 bb \rbrack_1\,1$. In this example, $n_1 = n_2 = 1$, $v_{[1]}=a$, $v_{[2]}=bb$, and
				\[
					\underline{\lbrack_1 a \rbrack_1}\,1\,\lbrack_1 bb \rbrack_1\,1 \to \lbrack_1 a \rbrack_1\,a\,\underline{\lbrack_1 bb \rbrack_1}\,1 \to \lbrack_1 a \rbrack_1\,a\,\lbrack_1 bb \rbrack_1\,bb \to aabbbb.
				\]
			\item $abc\,1\,2$. In this example, $n_1 = n_2 = 1$, $v_{[1]}=v_{[2]}=\varepsilon$, and $abc\,1\,2 \to abc\, 2 \to abc$.
		\end{enumerate}
	
    Note that an unbound reference is replaced with the empty string $\varepsilon$, that is, we adopt the $\varepsilon\text{-semantics}$. However, as mentioned in Section~\ref{sec:related}, this semantics' expressive power is equivalent to that of the $\emptyset\text{-semantics}$ under the ``may repeat labels'' convention (see \cite{berglund2023re} for the proof). We define the language $L(\alpha)$ denoted by a $k$-rewb $\alpha \in \rewb_k$ to be $\deref_k (\reflang_k(\alpha)) = \syuugou[\deref_k(v)]{v \in \reflang_k(\alpha)}$ (Lemmas~\ref{lem:deref} and \ref{lem:refw} ensure that $L(\alpha)$ is well-defined).
    
	Let $g: (\Sigma \mydisjointu B_k) ^ \ast \to \Sigma ^ \ast$ denote the free monoid homomorphism where $g(x)$ is $x$ for each $x \in \Sigma$, and $\varepsilon$ for each $x \in B_k$. Every $v \in (\Sigma \mydisjointu B_k \mydisjointu \slice{k}) ^ \ast$ can be written uniquely in the form $v = v_0 n_1 v_1 \cdots n_m v_m$, where $m \geq 0$ (denoted by $\cnt{v}$), and $v_r \in (\Sigma \mydisjointu B_k) ^ \ast$ and $n_r \in \slice{k}$ for each $r \in \{0,\dots,m\}$.
	Here, let $y_0 \triangleq v_0$ and for each $r \in \syuugou{1, \dots, m}$, $y_r \triangleq v_0 n_1 v_1 \cdots n_r v_r$. A string $v = v_0 n_1 v_1 \cdots n_m v_m$ over $\Sigma \mydisjointu B_k \mydisjointu \slice{k}$ is said to be \emph{matching} if
		\[
			\forall r \in \syuugou{1, \dots, m}.\, \forall x_1, x_2.\, y_{r-1} = x_1 \lbrack_ {n_r} x_2 \Longrightarrow (\exists x_2',x_3.x_2 = x_2 ^ \prime \rbrack_{n_r} x_3\,\land\,x_2^\prime \niton \lbrack_{n_r}, \rbrack_{n_r})
		\]
        holds. Intuitively, a string $v$ being matching means that for all $n_r \in \slice{k}$ in $v$, if there exists a left bracket $\lbrack_{n_r}$ in the string immediately before $n_r$, then there is a right bracket $\rbrack_{n_r}$ in between this $\lbrack_{n_r}$ and $n_r$.
    The following three lemmas follow.
	
	\begin{lemma} \label{lem:deref} Given a matching string $v$, $\deref_k(v) = g(v_0)\,g(v_{[1]})\,g(v_1) \cdots g(v_{[m]})\,g(v_m).$
	\end{lemma}

    \begin{lemma} \label{lem:prefix} A prefix of a matching string is matching. That is, if we decompose a string $v$ into $v = xy$, $x$ is matching. Moreover, $x_{[r]} = v_{[r]}$ holds for each $r = 1, \dots, \cnt{x}\,(\leq \cnt{v})$.
	\end{lemma}

	\begin{lemma} \label{lem:refw} Every $v \in \refwords _k$ is matching. 
	\end{lemma}

	Next, we recall the notions of SA, NESA, and NSA. In this paper, we unify their definitions based on \cite{aho1969nested, ginsburg1967one} to clarify the different capabilities of these models.
	First, we review NFA. Here is the definition in the textbook by Hopcroft et al.~\cite{hopcroft2001introduction}:
	
    \begin{definition2}[\cite{hopcroft2001introduction}, p.57] \label{dfn:nfa} A \emph{nondeterministic finite automaton} $N$ is a 5-tuple $(Q,\Sigma,\delta,q_0,F)$, where $Q$ is a finite set of states, $\Sigma$ a finite set of input symbols\,(also called alphabet), $q_0 \in Q$ a start state, $F \subseteq Q$ a set of final states, and $\delta: Q\times \Sigma \to \powerset{Q}$ a transition function.
	\end{definition2}
	
	As well known, the transition function $\delta$ can be \emph{extended} to $\hat{\delta}: Q \times \Sigma^\ast \to \powerset{Q}$ where $\hat{\delta}(q,w)$ represents the set of all states reachable from $q$ via $w$. Let $q \myto{a}{N} q^\prime$ denote $q^\prime \in \delta(q,a)$, and $q \mytto{w}{N} q^\prime$ denote $q^\prime \in \hat{\delta}(q,w)$. With this notation, the language of an NFA $N$ can be written as follows: $L(N) = \syuugou[w \in \Sigma^\ast]{\exists q_f \in F.\,q_0\mytto{w}{N} q_f}$.
	
    A pushdown automaton\,(PDA) is an NFA equipped with a \emph{stack} such that the PDA may write and read its \emph{stack top} with a transition. A \emph{stack automaton}\,(SA) is ``an extended PDA'', which can reference not only the top but inner content of the stack. That is, while the \emph{stack pointer} of a PDA is fixed to the top, an SA allows its pointer to move left and right and read a stack symbol pointed to by the pointer. However, the only place on the stack that can be rewritten is the top, as in PDA. Formally, a (1N)\,SA $A$ is a 9-tuple $(Q,\Sigma,\Gamma,\delta,q_0,Z_0,\#,\mydollar,F)$ satisfying the following conditions:
        the components $Q$, $\Sigma$, $q_0$ and $F$ are the same as those of NFA. $\Gamma\,(\neq \emptyset)$ is a finite set of stack symbols, and $Z_0 \in \Gamma$ is an initial stack symbol.		
                The stack symbol $\# \notin \Sigma \cup \Gamma$ (resp. $\mydollar \notin \Sigma \cup \Gamma$) is always and only written at the leftmost (bottom) (resp. the right most (top)) of the stack.\footnote{These special symbols $\#, \mydollar$ representing ``bottom'' and ``top'' of the stack respectively do not appear in \cite{ginsburg1967one} and are introduced anew in this paper to define NESA and NSA, which will be defined later, in the style of \cite{aho1969nested}. In fact, SA defined in \cite{ginsburg1967one} is not capable of directly discerning whether the stack pointer is at the top or not. Although it is not difficult to see that directly adding the ability does not increase the expressive power of SA, the ability is directly in NESA as seen in \cite{hopcroft1967nonerasing, ogden1969intercalation}. Therefore, to make it easy to see that NESA is a restriction of SA, we define SA to also directly have the ability.}
                The transition function $\delta$ has the following two modes, where $\myleft,\mystop,\myright \notin (\Sigma\cup\Gamma) \mydisjointu \syuugou{\#, \mydollar}$, $\dir_i \triangleq \syuugou{\mystop,\myright}$, and $\dir_s \triangleq \syuugou{\myleft,\mystop,\myright}$:
				\begin{romanenumerate}
					\item (pushdown mode) $Q\times \Sigma \times \Gamma\mydollar \to \powerset{Q\times \dir_i \times \Gamma^\ast \mydollar}$,
					\item (stack reading mode)
							(a) $Q \times \Sigma \times \Gamma\mydollar \to \powerset{Q\times \dir_i\times \syuugou{\myleft}}$,	
							(b) $Q \times \Sigma \times \Gamma \to \powerset{Q\times \dir_i\times \dir_s}$, 
							(c) $Q \times \Sigma \times \syuugou{\#} \to \powerset{Q\times \dir_i\times \syuugou{\myright}}$.
				\end{romanenumerate}
        
	Intuitively, $\delta$ works as follows (Definition~\ref{dfn:sa-id} provides the formal semantics).
            (i) The statement $(q^\prime, d, w\mydollar) \in \delta(q, a, Z\mydollar)$ says that whenever the current state is $q$, the input symbol is $a$, and the pointer references the top symbol $Z$, the machine can move to the state $q^\prime$, move the input cursor along $d$, and replace $Z$ with the string $w$.
            (ii) The statement (b) $(q^\prime,d,e) \in \delta(q,a,Z)$ says that whenever the current state is $q$, the input symbol is $a$, and the pointer references the symbol $Z$, the machine can move to the state $q^\prime$, move the input cursor along $d$, and move the pointer along $e$. The statements (a) and (c) are similar to (b) except that the direction in which the pointer can move is restricted lest the pointer go out of the stack. In particular, an SA that cannot erase a symbol once written on the stack is called a \emph{nonerasing stack automaton}\,(NESA). That is, 
	a (1N) nonerasing stack automaton is an SA whose transition function $\delta$ satisfies the condition that, in (i)\,(pushdown mode), $(q^\prime, d, w\mydollar) \in \delta(q,a,Z\mydollar)$ implies $w \in Z\Gamma^\ast$.
    To formally describe how SA works, we define a tuple called \emph{instantaneous description}\,(ID), which consists of a state, an input string, and a string representation of the stack, and define the binary relation $\myid{A}$ over the set of these tuples. Let $\mydtoz{\myleft}=-1$, $\mydtoz{\mystop}=0$, and $\mydtoz{\myright}=1$.

	\begin{definition2} \label{dfn:sa-id} Let $A$ be an SA $(Q,\Sigma,\Gamma,\delta,q_0,\#,\mydollar,F)$. An element of the set $I = Q \times \Sigma^\ast \times \syuugou{\#} (\Gamma \mydisjointu \syuugou{\,\myspl}) ^ \ast \syuugou{\mydollar}$ is called \emph{instantaneous description}, where the stack symbol $\myspl\,\notin \Gamma$ stands for the position of stack pointer. Moreover, let $\myid{A}$ (or $\myid{}$ when $A$ is clear) be the smallest binary relation over $I$ satisfying the following conditions:	
		\begin{romanenumerate}
			\item $(q,a_i\cdots a_k, \#yZ\myspl\!\mydollar) \myid{A} (q^\prime, a_{i+\mydtoz{d}} \cdots a_k, \#yw \myspl\!\mydollar)$ if $(q^\prime, d, w\mydollar) \in \delta(q,a_i,Z\mydollar)$.\footnote{We regard $a_{k+1}\cdots a_{k}$ as $\varepsilon$.}
			\item 
				\begin{alphaenumerate}
					\item $(q,a_i \cdots a_k, \#yZ\myspl\!\mydollar) \myid{A} (q^\prime,a_{i+\mydtoz{d}} \cdots a_k, \#y\myspl\!Z\mydollar)$ if $(q^\prime, d, \myleft) \in \delta(q,a_i,Z\mydollar)$.
					\item if $(q^\prime, d, e) \in \delta(q,a_i,Z)$ and $Z = Z_j$, $1 \leq j < n$, then
                                          
							\phantom{if }$(q,a_i \cdots a_k, \# Z_1 \cdots Z_{j} \myspl\! \cdots Z_n\mydollar) \myid{A} (q^\prime,a_{i+\mydtoz{d}} \cdots a_k, \# Z_1 \cdots Z_{j+\mydtoz{e}} \myspl\! \cdots Z_n\mydollar)$.
						
					\item $(q,a_i \cdots a_k, \#\myspl\! Zy\,\mydollar) \myid{A} (q^\prime,a_{i+\mydtoz{d}} \cdots a_k, \# Z \myspl\!y\,\mydollar)$ if $(q^\prime, d, \myright) \in \delta(q,a_i,\#)$.
				\end{alphaenumerate}
		\end{romanenumerate}
		
		Note that $\myleft \notin \dir_i$, which means the input cursor will not move back to left.
            We say that $A$ accepts $w \in \Sigma^\ast$ if there exist $y_1, y_2 \in \Gamma^\ast$, and $q_f \in F$ such that $(q_0, w, \# Z_0 \myspl\!\mydollar) \myid{A}^\ast (q_{f}, \varepsilon, \# y_1 \myspl\! y_2\, \mydollar)$. Let $L(A)$ denote the set of all strings accepted by $A$.
	\end{definition2}
	
	We next define \emph{nested stack automaton}\,(NSA) which is SA extended with the capability to create and remove substacks. For instance, suppose that the stack is $\# a_1 a_2\myspl\! a_3 \mydollar$ and we are to create a new substack containing $b_1 b_2$:
	\begin{equation}
		\underline{\# a_1 \mycent\,b_1 b_2 \myspl\! \mydollar}\,a_2 a_3 \mydollar \label{nsex-nested}.
	\end{equation}\par 
	
    Note that the new substack $\mycent\,b_1 b_2\,\mydollar$ is embedded below the symbol $a_2$ indicated by the stack pointer, and the pointer moves to the top of the created substack. The creation of the inner substack narrows the range within which the stack pointer can move as indicated by the underlined part $\# a_1 \mycent\,b_1 b_2 \myspl\! \mydollar$. While the bottom of the entire stack is always fixed by the leftmost symbol $\#$, the top of the embedded substack is regarded as the top of the entire stack. The inner substacks are allowed to be embedded endlessly and everywhere, whereas the writing in the pushdown mode is still restricted to the top of the stack:
	\begin{align}
		\underline{\# a_1 \mycent\,b_1 b_2 \myspl\! \mydollar}\,a_2 a_3 \mydollar
		&\overset{\myleft}{\rightarrow}
		\underline{\# a_1 \mycent\,b_1 \myspl\! b_2\,\mydollar}\,a_2 a_3 \mydollar
		\overset{\text{create}}{\longrightarrow}
		\underline{\# a_1 \mycent\, \mycent\, c_1 c_2 \myspl\! \mydollar}\,b_1 b_2 \mydollar\,a_2 a_3 \mydollar,
		\label{nsex-nestednested} \\
		\underline{\# a_1 \mycent \myspl\! b_1 b_2\,\mydollar}\,a_2 a_3 \mydollar
		&\overset{\myleft}{\rightarrow}
		\underline{\# a_1 \myspl\! \mycent\, b_1 b_2\,\mydollar}\,a_2 a_3 \mydollar
		\overset{\text{create}}{\longrightarrow}
		\underline{\# \mycent\, c_1 c_2 \myspl\! \mydollar}\,a_1 \mycent\, b_1 b_2\,\mydollar\,a_2 a_3 \mydollar. \label{nsex-undercent}
	\end{align}
    We must empty the inner substack and then remove itself in advance whenever we want to reference the right side of the inner substack such as $a_2, a_3$. For example, let us empty the inner substack by popping twice from \eqref{nsex-nested} and then removing it:
	\begin{equation}
		\underline{\# a_1 \mycent\,b_1 b_2 \myspl\! \mydollar}\,a_2 a_3 \mydollar
		\overset{\text{pop}}{\rightarrow}
		\underline{\# a_1 \mycent\,b_1 \myspl\! \mydollar}\,a_2 a_3 \mydollar
		\overset{\text{pop}}{\rightarrow}
		\underline{\# a_1 \mycent\myspl\! \mydollar}\,a_2 a_3 \mydollar
		\overset{\text{destruct}}{\longrightarrow}
		\underline{\# a_1 a_2\myspl\! a_3 \mydollar}.
		\label{nsex-delete}
	\end{equation}
    Notice that the stack pointer moves to the right after removing the inner substack. We now define NSA formally.
	A (1N) nested stack automaton $A$ is a 10-tuple $(Q,\Sigma,\Gamma,\delta,q_0,Z_0,\#,\mycent,\mydollar,F)$ satisfying the following conditions:
            the components $Q$, $\Sigma$, $\Gamma$, $q_0$, $Z_0$, $\#$, $\mydollar$ and $F$ are the same as those of SA.
			The stack symbol $\mycent \notin \Sigma \cup \Gamma$ represents the bottom of a substack.\footnote{Note that the bottom of the entire stack is always represented by $\#$ and not $\mycent$, as mentioned above.}
			The transition function $\delta$ has the following four modes, where $\Gamma ^ \prime \triangleq \Gamma \mydisjointu \syuugou{\mycent}$:
				\begin{romanenumerate}
					\item (pushdown mode) $Q\times \Sigma \times \Gamma\mydollar \to \powerset{Q\times \dir_i \times \Gamma^\ast \mydollar}$.
					\item (stack reading mode)
							(a) $Q \times \Sigma \times \Gamma^\prime\mydollar \to \powerset{Q\times \dir_i\times \syuugou{\myleft}}$,
							(b) $Q \times \Sigma \times \Gamma^\prime \to \powerset{Q\times \dir_i\times \dir_s}$, 
							(c) $Q \times \Sigma \times \syuugou{\#} \to \powerset{Q\times \dir_i\times \syuugou{\myright}}$.
					\item (stack creation mode) $Q \times \Sigma \times (\Gamma^\prime \mydisjointu \Gamma^\prime\mydollar) \to \powerset{Q\times \dir_i \times \syuugou{\mycent} \Gamma^\ast \mydollar}$.
					\item (stack destruction mode) $Q \times \Sigma \times \syuugou{\mycent}\mydollar \to \powerset{Q \times \dir_i}$.
				\end{romanenumerate}
	
    Moreover, we define how NSA works with ID and $\myid{}$ in the same manner as SA.
    Given an NSA $A = (Q,\Sigma,\Gamma,\delta,q_0,Z_0,\#,\mycent,\mydollar,F)$, we define ID, $\myid{A}$, and $L(A)$ in the same way as Definition~\ref{dfn:sa-id}\,(however, we let $I$ be $Q \times \Sigma^\ast \times \syuugou{\#}(\Gamma \mydisjointu \syuugou{\mycent, \mydollar, \,\myspl}) ^ \ast \syuugou{\mydollar}$). Here, we only give the rules corresponding to (iii) and (iv) in the definition of $\delta$ (the others are essentially the same as those of SA):
		\begin{romanenumerate}
			\setcounter{enumi}{2}
			\item if $(q^\prime,d,\mycent w\mydollar) \in \delta(q,a_i,Z)$ and $Z=Z_j$, $1\leq j < n$, then
				
					\phantom{if }$(q,a_i \cdots a_k, \#Z_1 \cdots Z_j \myspl\! \cdots Z_n \mydollar) \myid{A} (q^\prime,a_{i+\mydtoz{d}} \cdots a_k, \#Z_1 \cdots \mycent w \myspl\! \mydollar Z_j \cdots Z_n \mydollar)$,
				
				and $(q,a_i \cdots a_k, \# y Z \myspl\! \mydollar) \myid{A} (q^\prime,a_{i+\mydtoz{d}} \cdots a_k, \# y\,\mycent w \myspl\! \mydollar Z \mydollar)$ if $(q^\prime,d,\mycent w\mydollar) \in \delta(q,a_i,Z\mydollar)$.
			\item $(q,a_i \cdots a_k, \# y_1 \mycent \myspl\! \mydollar Z y_2 \mydollar) \myid{A} (q^\prime,a_{i+\mydtoz{d}} \cdots a_k, \# y_1 Z \myspl\! y_2 \mydollar)$ if $(q^\prime, d) \in \delta(q,a_i,\mycent \mydollar)$.
		\end{romanenumerate}
 \section{Every rewb describes an indexed language}
\label{sec:rewbisig}

As described above, to obtain the language $L(\alpha)$ described by a $k$-rewb $\alpha$, we derive the regular language $\reflang_k(\alpha)$ over the alphabet $\Sigma \mydisjointu B_k \mydisjointu \slice{k}$  first, then apply the dereferencing function $\deref_k$ to every element of $\reflang_k(\alpha)$. Using this observation, we construct an NSA $A_\alpha$ recognizing the language $L(\alpha)$ as follows.

The NSA $A_\alpha$ is based on an NFA $N$ recognizing the language $\reflang_k(\alpha)$, in the sense that each transition in $A_\alpha$ comes from a corresponding transition of $N$. The NFA $N$ has the alphabet $\Sigma \mydisjointu B_k \mydisjointu \slice{k}$, and so handles three types of characters. For each transition $q \myto{a}{N} q^\prime$ with $a \in \Sigma$, i.e., moving from $q$ to $q^\prime$ by an input symbol $a$, $A_\alpha$ also has the same transition except pushing $a$ to the stack, denoted by $q \myto{a/\mydollar \to a\mydollar}{} q^\prime$. 
        For each transition $q \myto{b}{N} q^\prime$ with $b \in B_k$, i.e., moving by a bracket $b$, $A_\alpha$ has the transition pushing $b$ without consuming input symbols, denoted by $q \myto{\varepsilon /\mydollar \to b\mydollar}{} q^\prime$.\footnote{Strictly speaking, our  NFA (cf.~Definition~\ref{dfn:nfa}) does not allow consuming the empty string $\varepsilon$. However, we can realize the transition $q \myto{\varepsilon /\mydollar \to b\mydollar}{} q^\prime$ alternatively by adding  $q \myto{c /\mydollar \to b\mydollar, \mystop}{} q^\prime$ for each $c \in \Sigma$, i.e., moving by $c$ with the input cursor fixed.}
        For each transition $q \myto{i}{N} q^\prime$ with $i \in \slice{k}$, $A_\alpha$ has a large ``transition'' that consists of several transitions. In this ``transition,'' $A_\alpha$ first seeks the left bracket $\lbrack_i$ of the bracketed string $\lbrack_i\, v \rbrack_i$ within the stack, and checks if the input from the cursor position matches $v$ character by character while consuming the input, and finally moves to $q^\prime$  if all characters of $v$ matched.

        A difficult yet interesting point is that NSA cannot check $v$ against the stack and push $v$ onto the stack at the same time, that is, after checking a character $c$ of $v$, if $A_\alpha$ wants to push $c$ to the stack, $A_\alpha$ must leave from $v$, climb up the stack toward the top, and write $c$. However, after the push, $A_\alpha$ becomes lost by not knowing where to go back to. How about marking the place where $A_\alpha$ should return in advance? Unfortunately, that does not work; NSA can insert such marks anywhere by creating substacks, but due to the restriction of NSA, it cannot go above the position of the mark, much less climb up to the top. Therefore, NSA cannot directly push the result of a dereference onto the stack.

        We cope with this problem as follows. We allow $j \in \slice{k}$ to appear in $v$, and for each appearance of $j$ in the checking of $v$, $A_\alpha$ pauses the checking and puts a substack containing the current state as a marker at the stack pointer position. Then, $A_\alpha$ searches down the stack for the corresponding bracketed string $\lbrack_j v^\prime \rbrack_j$, and begins checking $v^\prime$ if it is found. By repeating this process, $A_\alpha$ eventually reaches a string $v^{\prime\prime} \in (\Sigma \mydisjointu B_k) ^ \ast$ containing no characters of $\slice{k}$. Once done with the check of $v^{\prime\prime}$, $A_\alpha$ climbs up toward the stack top, finds a marker $p$ denoting the state to return to, and resumes from $p$ after deleting the substack containing the marker. By repeating this, if $A_\alpha$ returns to the position where it initially found $j$, it has successfully consumed the substring of the input string corresponding to the dereference of $j$.
        The following lemma is immediate.	
    \begin{lemma}\label{lem:basenfa} Let $k \geq 1$ and $\alpha \in \rewb _ k$. There exists an NFA $(Q,\Sigma \mydisjointu B_k \mydisjointu \slice{k},\delta,q_0,F)$ over $\Sigma \mydisjointu B_k \mydisjointu [k]$ recognizing $\reflang_k(\alpha)$ all of whose states can reach some final state, that is, $\forall q \in Q.\, \exists w \in (\Sigma \mydisjointu B_k \mydisjointu \slice{k})^\ast.\, \exists q_{f} \in F.\, q \mytto{w}{N} q_{f}$.
	\end{lemma}
	
    \begin{corollary} \label{cor:extend} Let $N$ be the NFA in Lemma~\ref{lem:basenfa}. For all $q \in Q$ and for all $w \in (\Sigma \mydisjointu B_k \mydisjointu \slice{k})^\ast$, if $q_0 \mytto{w}{N} q$ then $w$ is matching (see \fullversion{Appendix~\ref{app:extend}}{the full version~\cite{nogami2023expressive}} for the proof).
	\end{corollary}
    We show the main theorem (the proof sketch is coming later):
	\begin{theorem} \label{thm:main} 
        For every rewb $\alpha \in \rewb$, there exists an NSA that recognizes $L(\alpha)$.
	\end{theorem}
	
	The claim obviously holds when $\alpha$ is a pure regular expression (i.e., \,$\alpha \in \rewb_0$). Suppose that $\alpha \in \rewb_k$ with $k \geq 1$. By Lemma~\ref{lem:basenfa}, there is an NFA $N = (Q_N, \Sigma \mydisjointu B_k \mydisjointu [k], \delta_N, q_0, F)$ that recognizes $\reflang_k(\alpha)$ and satisfies Corollary~\ref{cor:extend}. We construct an NSA $A_\alpha = (Q,\Sigma,\Gamma,\delta,q_0,Z_0,\#,\mycent,\mydollar,F)$ as follows. Let $Q \triangleq Q_N \mydisjointu \syuugou[\symcall_i, \symexec_i, \symret_i]{i \in \slice{k}} \mydisjointu \syuugou[W_q]{q \in Q_N} \mydisjointu \syuugou[E_{p,i}, L_{p,i}]{p \in Q_N \mydisjointu \syuugou[\symexec_i]{i \in \slice{k}}, i \in \slice{k}}$, $\Gamma \triangleq \Sigma \mydisjointu B_k \mydisjointu \slice{k} \mydisjointu Q \mydisjointu \syuugou{Z_0}$, and let $\delta$ be the smallest relation that, for all $a \in \Sigma$, $b \in B_k$, $c \in \Sigma$, $i,j \in \slice{k}$, $q, q^\prime \in Q_N$, $Z \in \Gamma$ and $p \in Q_N \mydisjointu \syuugou[\symexec_i]{i \in \slice{k}}$, satisfies the following conditions:
    
    	\begin{figure}[htb]
    		\begin{minipage}[t]{0.535\linewidth}
		\begin{bracketenumerate} 
			\item $\delta_N(q,a) \ni q^\prime$ $\Longrightarrow$ $\delta(q,a,Z\mydollar) \ni (q^\prime, \myright, Za\mydollar)$ %
			\item $\delta_N(q,b) \ni q^\prime$ $\Longrightarrow$ $\delta(q,c,Z\mydollar) \ni (q^\prime, \mystop, Zb \mydollar)$ %
			\item $\delta_N(q,i) \ni q^\prime$ $\Longrightarrow$ $\delta(q,c,Z\mydollar) \ni (W_{q^\prime}, \mystop, Zi\mydollar)$ %
			\item $\delta(W_q,c,i\mydollar) = \syuugou{(\symcall_i,\mystop,\mycent q \mydollar)}$ %
			\item $\delta(\symcall_i,c,p \mydollar) = \syuugou{(\symcall_i,\mystop,\myleft)}$ %
			\item $\delta(\symcall_i,c,Z) = \syuugou{(\symcall_i, \mystop, \myleft)}$ where $Z \neq [_i, Z_0$ %
			\item $\delta(\symcall_i,c,Z_0) = \syuugou{(\symret_i, \mystop, \myright)}$ %
			\item $\delta(\symcall_i,c,[_i) = \syuugou{(\symexec_i,\mystop,\myright)}$ %

			\item $\delta(\symexec_i,a,a) = \syuugou{(\symexec_i,\myright,\myright)}$ %
		\end{bracketenumerate}
		\end{minipage}
		\begin{minipage}[t]{0.4505\linewidth}
        \begin{bracketenumerate}
			\setcounter{enumi}{9}
			\item $\delta(\symexec_i,c,[_j) = \syuugou{(\symexec_i,\mystop,\myright)}$ where $i \neq j$ %
			\item $\delta(\symexec_i,c,]_j) = \begin{cases}
 						\syuugou{(\symret_i,\mystop,\myright)}\; (i = j) \\
 						\syuugou{(\symexec_i,\mystop,\myright)}\; (i \neq j)
 			\end{cases}$ %
 			\item $\delta(\symexec_i,c,j) = \syuugou{(c_j,\mystop, \mycent \symexec_i \mydollar)}$ where $i \neq j$ %
			\item $\delta(\symret_i,c,Z) = \syuugou{(\symret_i,\mystop,\myright)}$
			\item $\delta(\symret_i,c,p\mydollar) = \syuugou{(E_{p,i},\mystop,\mydollar)}$
			\item $\delta(E_{p,i},c,\mycent\mydollar) = \syuugou{(L_{p,i},\mystop)}$
			\item $\delta(L_{\symexec_j,i},c,i) = \syuugou{(\symexec_j,\mystop,\myright)}$
			\item $\delta(L_{q,i},c,i\mydollar) = \syuugou{(q,\mystop,\mystop)}$
		\end{bracketenumerate}
		\end{minipage}
		\end{figure}
		
     Rule (1) translates $q \myto{a}{N} q^\prime$ into $q \myto{a/\mydollar\to a\mydollar}{} q^\prime$, (2) translates $q \myto{b}{N} q^\prime$ into $q \myto{\varepsilon/\mydollar\to b \mydollar}{} q^\prime$, and rules (3)--(17) translate $q \myto{i}{N} q^\prime$ into a large ``transition'' to consume the string that corresponds to the dereference of $i$. The details of the ``transition'' are as follows.
    By looking at the underlying $N$ with rule (3), $A_\alpha$ finds a state $q^\prime$ that it should go back to after going throughout the ``transition,'' and goes to the state $W_{q^\prime}$ by pushing $i$ to the stack. At $W_{q^\prime}$, by rule (4), $A_\alpha$ inserts $\mycent q^\prime \mydollar$ just below $i$, and goes to the state $\symcall_i$. The state $\symcall_i$ represents the \emph{call mode} in which $A_\alpha$ looks for the left-nearest $\lbrack_i$ by rules (5) and (6) and proceeds to the state $\symexec_i$ (\emph{execution mode}) by (8) if it finds $\lbrack_i$. Otherwise (i.e., the case when $A_\alpha$ arrives at the bottom of the stack), it proceeds to the state $\symret_i$ (\emph{return mode}) by rule (7). At $\symexec_i$, $A_\alpha$ consumes input symbols by checking them against the symbols on the stack (rules (9)--(12)). In particular, rule (9) handles the case when the symbols match. Rules (10) and (11) handle the cases when brackets are read from the stack. The first case of (11) handles the case when the right bracket $\rbrack_i$ is read, and the rules handle the other brackets (i.e., $\lbrack_j$ or $\rbrack_j$ with $i \neq j$) by simply skipping them (note that $\lbrack_j = \mathord{\lbrack_i}$ cannot happen since we started from the left-nearest $\lbrack_i$). Reading $j \in \slice{k}$, by rule (12), $A_\alpha$ inserts $\mycent \symexec_i \mydollar$ just below $j$ and goes to $\symcall_j$ to locate the corresponding $\lbrack_j$ (here, $j \neq i$ holds by the definition of the syntax). At $\symret_i$, $A_\alpha$ proceeds to return to the state $p$ that passed the control to $\symcall_i$ (rules (13)--(17)). Since this $p$ was pushed at the stack top, $A_\alpha$ first climbs up to the stack top by rule (13), transits to the state $E_{p,i}$ popping $p$ by (14), then goes to $L_{p,i}$ removing the embedded substack by (15), and finally goes back to $p$ by (16) and (17). A subtle point in the last step is that where the stack pointer should be placed depends on whether $p$ is a state $\symexec_j$ (for some $j \in \slice{k}$) or in $Q_N$. In the former case, after (15) removes the embedded substack $\mycent \symexec_j \mydollar$ that was created just below the call to $i$, the stack pointer points to $i$. However, the stack pointer should shift one more to the right, lest $A_\alpha$ begins to repeat the call reading $i$ again by (12). Therefore, (16) correctly handles the case by doing the shift. In the latter case, as stipulated by (17), the stack pointer should point to the stack top symbol $i$ since $p$ is the state stored at (3).
	
	We state two lemmas used to prove Theorem~\ref{thm:main}. Let $\myid{(n)}$ denote the subrelation of $\myid{}$ derived from the rule $(n)$. The following lemma is immediate from the definition of $\vdash_{(n)}$.

	\begin{lemma} \label{lem:12}
        For all $q, q^\prime \in Q_N$, $w, w^\prime \in \Sigma ^ \ast$, $\gamma, \gamma^\prime \in \Gamma ^ \ast$, 
		\begin{alphaenumerate}
			\item \begin{enumerate}[1.]
 				\item for each $a \in \Sigma$, $(q, aw, \# Z_0 \gamma \myspl\! \mydollar) \vdash_{(1)} (q^\prime, w, \# Z_0 \gamma a \myspl\! \mydollar)$ if $q \myto{a}{N} q^\prime$,
 				\item $\exists a \in \Sigma.\,q\myto{a}{N}q^\prime \land w = aw^\prime \land \beta = Z_0 \gamma a \myspl$ if $(q,w, \# Z_0 \gamma \myspl\! \mydollar) \vdash_{(1)} (q^\prime, w^\prime, \# \beta \mydollar)$,
 			\end{enumerate}
			\item \begin{enumerate}[1.]
 				\item for each $b \in B_k$, $(q, w, \# Z_0 \gamma \myspl\! \mydollar) \vdash_{(2)} (q^\prime, w, \# Z_0 \gamma b \myspl\! \mydollar)$ if $q \myto{b}{N} q^\prime$,
 				\item $\exists b \in B_k.\,q\myto{b}{N}q^\prime \land w = w^\prime \land \beta = Z_0 \gamma b \myspl$ if $(q,w, \# Z_0 \gamma \myspl\! \mydollar) \vdash_{(2)} (q^\prime, w^\prime, \# \beta \mydollar)$.
 			\end{enumerate}
		\end{alphaenumerate}	
	In particular, letting $\mathord{\myid{(1),(2)}} = \mathord{\myid{(1)}} \mydisjointu \mathord{\myid{(2)}}$, we obtain the following statement by repeating (a)$_1$ and (b)$_1$ zero or more times:
For all $v \in (\Sigma \mydisjointu B_k) ^ \ast$, $(q,g(v)\,w, \#Z_0 \gamma \myspl\! \mydollar) \myid{(1),(2)}^\ast (q^\prime, w, \# Z_0 \gamma v \myspl\! \mydollar)$ if $q \mytto{v}{N} q^\prime$.

	\end{lemma}

	\begin{lemma} \label{lem:3}
        Suppose that $q \underset{N}{\overset{i}{\longrightarrow}} q^\prime$, and $\gamma i$ is matching. Let $m = \cnt{(\gamma i)}$. For all $p \in Q_N$, $w,w ^ \prime \in \Sigma^\ast$ and $\beta \in (\Gamma \mydisjointu \syuugou{\mycent, \mydollar,\,\myspl})^\ast$, the following (a) and (b) are equivalent (see Appendix~\ref{app:3} for the proof):
		\begin{alphaenumerate}
			\item $p = q^\prime$, $w = g((\gamma i)_{[m]})\,w^\prime$, and $\beta = Z_0 \gamma i \myspl$.
            \item $(q, w, \# Z_0 \gamma \myspl\!\mydollar) \myid{(3)} (W_{q^\prime}, w, \# Z_0 \gamma i \myspl\! \mydollar) \myid{} \cdots \myid{} (p, w^\prime, \# \beta \mydollar)$, where no ID with a state in $Q_N$ appears in the calculation $\cdots$.
		\end{alphaenumerate}
	\end{lemma}

	\begin{proof}[Proof of Theorem~\ref{thm:main}\,(sketch)]
        For proving $L(\alpha) \subseteq L(A_\alpha)$, we take $w \in L(\alpha)$ and $v \in \reflang_k(\alpha)$ such that $w = \deref_k(v)$. Decomposing $v$ into $v_0 n_1 v_1 \cdots n_m v_m$ (where $m = \cnt{v}$), we obtain a transition sequence in the underlying NFA $N$, denoted by $q_0 \underset{N}{\overset{v_0}{\Longrightarrow}} q_{(0)} \underset{N}{\overset{n_1 v_1}{\Longrightarrow}} q_{(1)} \underset{N}{\overset{n_2 v_2}{\Longrightarrow}} \cdots \underset{N}{\overset{n_m v_m}{\Longrightarrow}} q_{(m)} \in F$. We prove by induction on $r = 0,\dots,m$ that $A_\alpha$ can reach $q_{(r)}$ while consuming $z_r = g(v_0)\,g(v_{[1]})\,g(v_1) \cdots g(v_{[r]})\,g(v_r)$ from the input and pushing $y_r = v_0 n_1 v_1 \cdots n_r v_r$ to the stack. Conversely, we suppose a calculation in $A_\alpha$, denoted by $C_{(1)} = (q_0,w, \# Z_0 \myspl\! \mydollar) \vdash \cdots \vdash C_{(r)} \vdash \cdots \vdash C_{(m)} = (p_m, \varepsilon, \# \beta_m \mydollar)$, where $p_m \in F$ and $C_{(r)} = (p_r, w_r, \# \beta _ {r} \mydollar)$ for each $r \in \syuugou{1,\dots,m}$. By induction on $r=1,\dots,m$, we extract an underlying transition $q_0 \underset{N}{\overset{\gamma_r}{\Longrightarrow}} p_r$ step by step while maintaining the invariants $\gamma_r \in (\Sigma \mydisjointu B_k \mydisjointu \slice{k}) ^ \ast$ and $w = \deref_k(\gamma_r)\,w_r$, as long as $p_{r} \in Q_N$ (the formal proof is available in \fullversion{Appendix~\ref{app:main}}{the full version~\cite{nogami2023expressive}}).
	\end{proof}
	
	\begin{corollary} \label{cor:rewbisig}
	Every rewb describes an indexed language, but not vice versa.
	\end{corollary}
	
	\begin{proof}
        The first half follows by Theorem~\ref{thm:main} since 1N NSA and indexed grammars are equivalent~\cite{aho1969nested}. The second half also follows since the class of CFLs is a subclass of indexed languages~\cite{aho1968indexed}, and the class of rewbs and that of CFLs are incomparable~\cite{berglund2023re}.
    \end{proof}

    In the case of a rewb $\alpha$ without a captured reference (that is, one in which no reference $\bs i$ appears as a subexpression of an expression of the form $(_j \dots )_j$), we can transform $A_\alpha$ into an NESA $A_\alpha^{\prime\prime}$ recognizing $L(\alpha)$, i.e., one that neither uses substacks nor pops its stack. First, we transform $A_\alpha$ to an NSA without substacks (i.e., SA) $A_\alpha^\prime$. Inspecting how substacks are used in $A_\alpha$, we can drop rules (12) and (16) in $A_\alpha^\prime$ because there is no captured reference in $\alpha$. We also remove the uses of substacks from rules (3) and (4), which correspond to calling, and rules (14), (15) and (17), which correspond to returning. Namely, while $A_\alpha$, upon a call, stores the substack $\mycent q^\prime \mydollar$ that consists of just the state $q^\prime$ where the control should return, $A_\alpha^\prime$ simply pushes $q^\prime$ to the stack top. That is, we remove (4), (15) and (17), and change (3) and (14) to the following (3') and (14'), respectively:
	\[
		\text{(3')}\ \delta_N(q,i) \ni q^\prime \Longrightarrow \delta(q,c,Z\mydollar) \ni (\symcall_i,\mystop, Ziq^\prime\mydollar), \qquad \text{(14')}\ \delta(\symret_i,c,q\mydollar) = \syuugou{(q,\mystop,\mydollar)}.
	\]
    Furthermore, we transform $A_\alpha^\prime$ to an SA without stack popping (i.e., NESA) $A_\alpha^{\prime\prime}$. Observe that $A_\alpha^\prime$ pops only when returning via (14') and popping a state that was pushed in a preceding call. Thus, $A_\alpha^{\prime\prime}$, rather than popping $q^\prime$, leaves it on the stack, and has the modes $\symcall_i$, $\symexec_i$ and $\symret_i$ skip all state symbols on the stack except the ones at the top. Here, we only need to modify $\symexec_i$ since $A_\alpha$ already skips them at $\symcall_i$ and $\symret_i$ (rules (6) and (13)). In short, we add the new rule (9*) and change (14') to (14''), as follows:
	\[
		\text{(9*)}\ \delta(\symexec_i,c,q) = \syuugou{(\symexec_i,\mystop,\myright)}, \qquad \qquad \text{(14'')}\ \delta(\symret_i,c,q\mydollar) = \syuugou{(q,\mystop,q\mydollar)}.
	\]
	This NESA $A_\alpha^{\prime\prime}$ whose transition function consists of the rules (1),(2),(3'),(5)--(9),(9*),(10),\\ 
    (11), (13) and (14'') recognizes $L(\alpha)$. Therefore, 
	
	\begin{corollary} \label{cor:refcap}
        Every rewb without a captured reference describes a nonerasing stack language, but not vice versa.\footnote{For the latter part, we can take the language $\syuugou[a^n b^n]{n \in \mynat}$ that can be described by an NESA (see \fullversion{Appendix~\ref{app:anbnNESA}}{the full version~\cite{nogami2023expressive}}) but not by any rewb~\cite{berglund2023re}.}
	\end{corollary}

    Note that the converse of Corollary~\ref{cor:refcap} fails to hold. In other words, there is a rewb with a captured reference that describes a nonerasing stack language. The rewb $(_1 a )_1 (_2 \bs 1 )_2 \bs 2$ is a simple counterexample. In addition, as shown later in Section~\ref{sec:larsen98isnonerasingsl}, NESA can recognize nontrivial language (hierarchy) with a captured reference such as Larsen's hierarchy~\cite{larsen1998regular}.
 \section{A rewb that describes a non-stack language}
\label{sec:rewbisnotsl}

We just showed that every rewb describes an indexed language and in particular every rewb without a captured reference describes a nonerasing stack language. So, a natural question is whether every rewb describes a (nonerasing) stack language. We show that the answer is \emph{no}. That is, there is a rewb that describes a non-stack language.

Ogden has proposed a pumping lemma for stack languages and shown that the language $\syuugou[a^{n^3}]{n \in \mynat}$ is a non-stack language as an application (see \cite{ogden1969intercalation}, Theorem~2). A key point in the proof is that the exponential $n^3$ of $a$ is a cubic polynomial, and we can show that for every cubic polynomial $f: \mynat \to \mynat$, the language $\syuugou[a^{f(n)}]{n \in \mynat}$ is also a non-stack language by the same proof. Thus, a rewb that describes a language in this form is a counterexample. 
We borrow the technique in \cite{freydenberger2019deterministic} (Example~1) which shows that the rewb $\alpha = ((_1 \bs 2 )_1 (_2 \bs 1 a )_2 )^\ast$ describes $L(\alpha) = \syuugou[a^{n^2}]{n \in \mynat}$. This follows since $\deref_k((\lbrack_12\rbrack_1\,\lbrack_2 1\,a\rbrack_2)^n) = a^{n^2}$ holds by recording the iteration count of the Kleene star, $n$, in the capture $(_2\,)_2$ as $a^n$, and extending the length by $2n+1$, as shown below:
	\begin{align*}
		&\deref_k((\lbrack_12\rbrack_1\,\lbrack_2 1\,a\rbrack_2)^{n+1}) = \deref_k((\lbrack_12\rbrack_1\,\lbrack_2 1\,a\rbrack_2)^n \lbrack_12\rbrack_1\,\lbrack_2 1\,a\rbrack_2) = \deref_k(\cdots \underline{\lbrack_2 a^n \rbrack_2}\, \lbrack_12\rbrack_1\,\lbrack_2 1\,a\rbrack_2 ) \\
        &= \deref_k(\cdots \lbrack_2 a^n \rbrack_2\, \underline{\lbrack_1 a^n \rbrack_1}\,\lbrack_2 1\,a\rbrack_2 ) = \deref_k(\uuline{\cdots \lbrack_2 a^n \rbrack_2}\, \lbrack_1 a^n \rbrack_1\,\lbrack_2 a^{n+1}\rbrack_2 ) = \uuline{\vphantom{\lbrack_2}a^{n^2}} \, a^{2n+1} = a^{(n+1)^2}.
	\end{align*}
        The rewb $((_1 \bs 4\,a )_1\,(_2 \bs 3 )_2\,(_3 \bs 2\,a)_3\,(_4 \bs 1 \bs 3 )_4)^\ast$ describes \syuugou[a^{n(n+7)(2n+1)/6}]{n \in \mynat} and extends the length by a quadratic in $n$ instead (see \fullversion{Appendix~\ref{app:calc}}{the full version~\cite{nogami2023expressive}} for the calculation). Thus,
	\begin{theorem} \label{thm:antisl} There exists a rewb that describes a non-stack language.
	\end{theorem}
    From this and Corollary~\ref{cor:refcap}, this rewb needs a captured reference, in the sense that:
    \begin{corollary} \label{cor:needrefcap} There exists a rewb that describes a language that no rewb without a captured reference can describe.
	\end{corollary}
 \section{Larsen's hierarchy is within the class of nonerasing stack language}
\label{sec:larsen98isnonerasingsl}
In this section, we construct an NESA $A_i$ that describes $L(x_i)$, where the rewb $x_i$ over the alphabet $\Sigma = \syuugou{a_0^l, a_0^m, a_0^r, a_1^l, a_1^m, a_1^r, \dots}$ is given by Larsen~\cite{larsen1998regular} and defined as follows: $x_0 \triangleq (a_0^l a_0^m a_0^r)^\ast$, $x_{i+1} \triangleq (a_{i+1}^l (_{i}\,x_{i}\,)_{i}\, a_{i+1}^m \bs i\,a_{i+1}^r) ^ \ast$ $(i \geq 0)$. Our result implies that Larsen's hierarchy is within the class of nonerasing stack languages.
Since Larsen showed that no rewb with its nested level less than $i$ can describe $L(x_i)$~\cite{larsen1998regular}, it also implies that for every $i \in \mathbb{N}$, there is a nonerasing stack language that needs a rewb of nested level at least $i$.\footnote{Technically, Larsen~\cite{larsen1998regular} adopts a syntax that excludes unbound references, and so this implied result applies only to rewbs with no unbound references.}
	\begin{figure}[htb] \centering
			\begin{minipage}{0.25\linewidth}
				\begin{tikzpicture}[
						shorten >=1pt,
						node distance=2.25cm,
						on grid,
						>={Stealth[round]},
						auto
					]
					\tikzset{el/.style = {inner sep=2pt, align=left, sloped}}
					\tikzset{every node/.style={font=\small}}
					\tikzset{every loop/.style={min distance=6mm,looseness=8}}
					\tikzset{every state/.style={inner sep=.1mm, minimum size=5mm}}
					\tikzset{initial text=}
					\tikzset{every initial by arrow/.style={->}}
					\tikzset{bend angle=30}
				
					\node[state,initial,accepting] (7) at (0.7,-4) {$q_0^0$};
					\path[->] (7) edge [in=-60,out=-120,loop below] node {\scriptsize
									$a_0^l a_0^m a_0^r/ \figto{\mydollar}{a_0^l a_0^m a_0^r \mydollar}$
								} ();
				\end{tikzpicture}
			\end{minipage}
			\begin{minipage}{0.52\linewidth}
				\begin{tikzpicture}[
						shorten >=1pt,
						node distance=2.25cm,
						on grid,
						>={Stealth[round]},
						auto
					]
					\tikzset{el/.style = {inner sep=2pt, align=left, sloped}}
					\tikzset{every node/.style={font=\small}}
					\tikzset{every loop/.style={min distance=6mm,looseness=8}}
					\tikzset{every state/.style={inner sep=.1mm, minimum size=5mm}}
					\tikzset{initial text=}
					\tikzset{every initial by arrow/.style={->}}
					\tikzset{bend angle=30}
			
					\node[state,initial,accepting] (2) at (0.6,-2) {$q_0^1$};
					\node[state] (6) at (0.1,-3) {};
					\node[state] (7) at (0.7,-4) {$q_0^0$};
					\node[state] (8) at (1.4,-3) {};
					\node[state] (9) at (3.4,-3) {};
					\node[state] (10) at (4.4,-3) {};
					\node[state] (13) at (3.4,-4.) {$\symcall_0^1$};
					\node[state] (14) at (4.4,-4.) {$\symret_0^1$};
					\node[state] (16) at (3.9,-5.) {$\symexec_0^1$};
					
					\path[->] (2) edge node [left] {\scriptsize
									$a_1^l/\figto{\mydollar}{a_1^l \mydollar}$
								} (6);
					\path[->] (6) edge node [left] {\scriptsize
									$\figto{\mydollar}{\lbrack_0\mydollar}$
								} (7);
					\path[->] (7) edge [in=-60,out=-120,loop below] node {\scriptsize
									$a_0^l a_0^m a_0^r/ \figto{\mydollar}{a_0^l a_0^m a_0^r \mydollar}$
								} ();
					\path[->] (7) edge node [right] {\scriptsize
									$\figto{\mydollar}{\rbrack_0\mydollar}$
								} (8);
					\path[->] (8) edge node {\scriptsize
									$a_1^m/\figto{\mydollar}{a_1^m\mydollar}$
								} (9);
					\path[->] (9) edge node [left] {\scriptsize
									$\figto{\mydollar}{0\mydollar}$
								} (13);
					\path[->] (10) edge [bend right=10] node [above] {\scriptsize
									$a_1^r/\figto{\mydollar}{a_1^r\mydollar}$
								} (2);
					\path[->] (13) edge node [left] {\scriptsize
									$\lbrack_0,\myright$
								} (16);
					\path[->] (13) edge [in=-60,out=-120,loop left] node {\scriptsize
									$\lnot \lbrack_0,\myright$
								} ();
					\path[->] (14) edge node {\scriptsize
									$0\mydollar,\mystop$
								} (10);
					\path[->] (14) edge [in=30,out=60, loop right] node [right] {\scriptsize
									$\lnot 0, \myright$
								} ();
					\path[->] (16) edge [in=-60,out=-120,loop left] node {\scriptsize
									$a/a,\myright$
								} ();
					\path[->] (16) edge node [right] {\scriptsize
									$\rbrack_0,\myright$
								} (14);
				\end{tikzpicture}
			\end{minipage} \par \medskip
			
			\qquad
			\begin{minipage}[b]{0.84\linewidth}
				\begin{tikzpicture}[
						shorten >=1pt,
						node distance=2.25cm,
						on grid,
						>={Stealth[round]},
						auto
					]
					\tikzset{el/.style = {inner sep=2pt, align=left, sloped}}
					\tikzset{every node/.style={font=\small}}
					\tikzset{every loop/.style={min distance=6mm,looseness=8}}
					\tikzset{every state/.style={inner sep=.1mm, minimum size=5mm}}
					\tikzset{initial text=}
					\tikzset{every initial by arrow/.style={->}}
					\tikzset{bend angle=30}
					\node[state,initial,accepting]  (0)  {$q_0^2$};
					\node[state] (1) at (-0.5,-1) {};
					\node[state] (2) at (0.6,-2) {$q_0^1$};
					\node[state] (3) at (2,-1) {};
					\node[state] (4) at (6.2,-1) {};
					\node[state] (5) at (7.4,-1) {};
					\node[state] (6) at (0.1,-3) {};
					\node[state] (7) at (0.7,-4) {$q_0^0$};
					\node[state] (8) at (1.4,-3) {};
					\node[state] (9) at (3.4,-3) {};
					\node[state] (10) at (4.4,-3) {};
					\node[state] (11) at (6.2,-2.) {$\symcall_1^2$};
					\node[state] (12) at (7.4,-2.) {$\symret_1^2$};
					\node[state] (13) at (3.4,-4.) {$\symcall_0^1$};
					\node[state] (14) at (4.4,-4.) {$\symret_0^1$};
					\node[state] (15) at (6.8,-3.) {$\symexec_1^2$};
					\node[state] (16) at (3.9,-5.) {$\symexec_0^1$};
					\node[state] (17) at (6.2,-4.) {$\symcall_0^2$};
					\node[state] (18) at (7.4,-4.) {$\symret_0^2$};
					\node[state] (19) at (6.8,-5.) {$\symexec_0^2$};
					
					\path[->] (0) edge node [left] {\scriptsize
									$a_2^l/\figto{\mydollar}{a_2^l \mydollar}$
								} (1);
					\path[->] (1) edge node [left] {\scriptsize
									$\figto{\mydollar}{\lbrack_1 \mydollar}$
								} (2);
					\path[->] (2) edge node [right] {\scriptsize
									$\figto{\mydollar}{\rbrack_1 \mydollar}$
								} (3);
					\path[->] (2) edge node [left] {\scriptsize
									$a_1^l/\figto{\mydollar}{a_1^l \mydollar}$
								} (6);
					\path[->] (3) edge node {\scriptsize
									$a_2^m/\figto{\mydollar}{a_2^m\mydollar}$
								} (4);
					\path[->] (4) edge node [left] {\scriptsize
									$\figto{\mydollar}{1\mydollar}$
								} (11);
					\path[->] (5) edge [bend right=10] node [above] {\scriptsize
									$a_2^r/\figto{\mydollar}{a_2^r\mydollar}$
								} (0);
					\path[->] (6) edge node [left] {\scriptsize
									$\figto{\mydollar}{\lbrack_0\mydollar}$
								} (7);
					\path[->] (7) edge [in=-60,out=-120,loop below] node {\scriptsize
									$a_0^l a_0^m a_0^r/ \figto{\mydollar}{a_0^l a_0^m a_0^r \mydollar}$
								} ();
					\path[->] (7) edge node [right] {\scriptsize
									$\figto{\mydollar}{\rbrack_0\mydollar}$
								} (8);
					\path[->] (8) edge node {\scriptsize
									$a_1^m/\figto{\mydollar}{a_1^m\mydollar}$
								} (9);
					\path[->] (9) edge node [left] {\scriptsize
									$\figto{\mydollar}{0\mydollar}$
								} (13);
					\path[->] (10) edge [bend right=10] node [above] {\scriptsize
									$a_1^r/\figto{\mydollar}{a_1^r\mydollar}$
								} (2);
					\path[->] (11) edge node [left] {\scriptsize
									$\lbrack_1,\myright$
								} (15);
					\path[->] (11) edge [in=-60,out=-120,loop left] node {\scriptsize
									$\lnot \lbrack_1, \myleft$
								} ();
					\path[->] (12) edge node {\scriptsize
									$1\mydollar,\mystop$
								} (5);
					\path[->] (12) edge [in=-60,out=-120,loop right] node {\scriptsize
									$\lnot 1, \myright$
								} ();
					\path[->] (13) edge node [left] {\scriptsize
									$\lbrack_0,\myright$
								} (16);
					\path[->] (13) edge [in=-60,out=-120,loop left] node {\scriptsize
									$\lnot \lbrack_0,\myright$
								} ();
					\path[->] (14) edge node {\scriptsize
									$0\mydollar,\mystop$
								} (10);
					\path[->] (14) edge [in=30,out=60, loop right] node [below] {\scriptsize
									$\lnot 0, \myright$
								} ();
					\path[->] (15) edge node [left] {\scriptsize
									$0,\myleft$
								} (17);
					\path[->] (15) edge node [right] {\scriptsize
									$\rbrack_1,\myright$
								} (12);
					\path[->] (15) edge [in=-60,out=-120,loop left] node {\scriptsize
									$a/a,\myright$
								} ();
					\path[->] (15) edge [in=-60,out=-120,loop right] node {\scriptsize
									$\syuugou{\lbrack_0,\rbrack_0},\myright$
								} ();
					\path[->] (16) edge [in=-60,out=-120,loop left] node {\scriptsize
									$a/a,\myright$
								} ();
					\path[->] (16) edge node [right] {\scriptsize
									$\rbrack_0,\myright$
								} (14);
					\path[->] (17) edge node [left] {\scriptsize
									$\lbrack_0,\myright$
								} (19);
					\path[->] (17) edge [in=175,out=195,loop left] node [above] {\scriptsize
									$\lnot \lbrack_0,\myleft$
								} ();
					\path[->] (18) edge node [right] {\scriptsize
									$0,\myright$
								} (15);
					\path[->] (18) edge [in=-60,out=-120,loop right] node {\scriptsize
									$\lnot 0, \myright$
								} ();
					\path[->] (19) edge node [right] {\scriptsize
									$\rbrack_0,\myright$
								} (18);
					\path[->] (19) edge [in=-60,out=-120,loop left] node {\scriptsize
									$a/a,\myright$
								} ();
				\end{tikzpicture}
			\end{minipage}
		
		\caption{$A_0$~(upper left),\ $A_1$~(upper right),\ $A_2$~(lower)}
		\label{fig:A0A1A2}
	\end{figure}
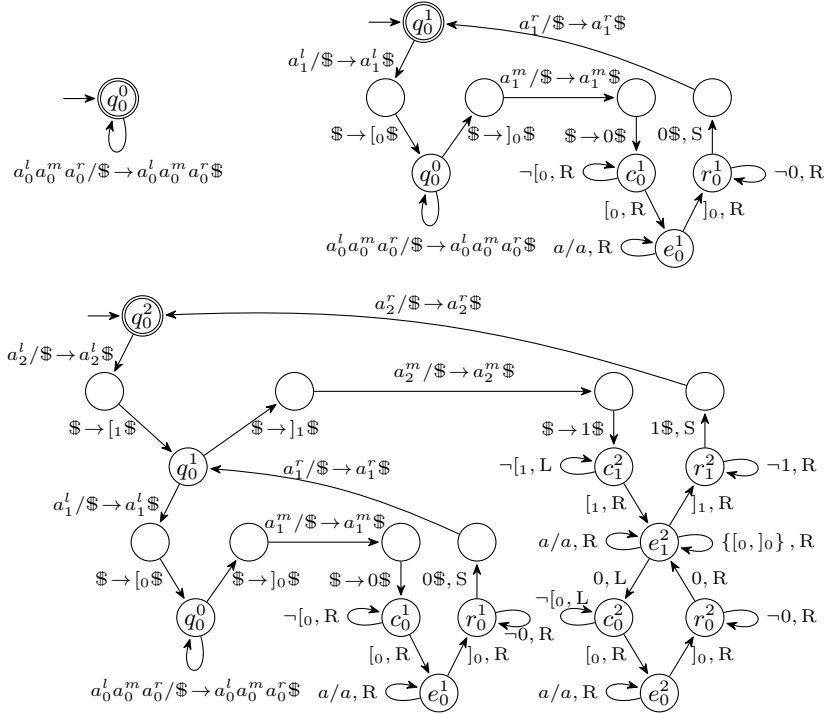
        
    The NESA $A_i$ has the start state $q_0^i$ which is also its only final state. Figure~\ref{fig:A0A1A2} depicts $A_0$, $A_1$, and $A_2$.
    $A_0$ is easy. $A_1$ is obtained by connecting the eight states to $q_0^0$ and making $q_0^1$ the start/final state, as shown in the figure. The five states on the right handle the dereference of $\bs 0$ in $x_1$. That is, at $\symcall_0^1$, $A_1$ first seeks the left-nearest $\lbrack_0$, passes the control to $\symexec_0^1$, checks the input string against the stack at $\symexec_0^1$, passes the control to $\symret_0^1$, and at $\symret_0^1$, finally goes back to the right-nearest $0$ which must be written on the stack top. In much the same way, $A_2$ is obtained from $A_1$ but we must be sensitive to the handling of the dereference of $\bs 1$ because $A_2$ must handle the dereference of not only $\bs 1$ but also $\bs 0$ that appears in a string captured by $\lbrack_1 \rbrack_1$ whereas no backreference appears in a string captured by $\lbrack_0 \rbrack_0$ in the case of $A_1$.
    To deal with this issue, we connect the three new states $\symcall_0^2$, $\symexec_0^2$ and $\symret_0^2$ to $\symexec_1^2$. At $\symexec_1^2$, if $A_2$ encounters $0$ in a checking, $A_2$ suspends the checking and first goes to $\symcall_0^2$ to seek $\lbrack_0$, goes to $\symexec_0^2$ to check the input against the stack by reading out a $\rbrack_0$ (no number appears in this checking), and finally goes to $\symret_0^2$ to go back to $0$ which passed the control to $\symcall_0^2$. We repeat this modification until $A_i$ is obtained. (Thus, $A_i$ has such states $\symcall_j^i, \symexec_j^i, \symret_j^i$ for each $j \in \syuugou{0,\dots,i-1}$.) Therefore, 
	
	\begin{theorem} \label{thm:hier} There exists an NESA $A_i$ that recognizes $L(x_i)$.
	\end{theorem}
 \section{Conclusions}
\label{sec:conc}
	\begin{figure}[htb] \centering
		\begin{tikzcd}[column sep={between origins,7em}, row sep={between origins,2.8em}]
			\text{REG} \arrow[rr, twoheadrightarrow] \arrow[rdd, twoheadrightarrow] & [-45pt] & \text{CFL} \arrow[rr, twoheadrightarrow] & [-65pt] & \text{SL} \arrow[rr, "\text{\cite{aho1969nested}}", twoheadrightarrow] & [-20pt] & [-62pt] \text{IL} \arrow[r, "\text{\cite{aho1968indexed}}", twoheadrightarrow] & \text{CSL} \\
			& & & \text{NESL} \arrow[ru, swap, "\text{\cite{ogden1969intercalation}}", near end, twoheadrightarrow] \\ [23pt]
            & \substack{\text{rewb without} \\ \text{a captured reference}} \arrow[rru, mydeepred, dashed, "\substack{\text{Corollary} \\ \text{\ref{cor:refcap}}}", pos=0.85, twoheadrightarrow] \arrow[rrrr, shift right] \arrow[from=rrrr, shift right=1.7, "/" marking, mydeepred, dashed] \arrow[from=rrrr, phantom, shift right=4.2, "\text{\scriptsize Corollary~\ref{cor:needrefcap}}", mydeepred, pos=0.3]  \arrow[ruu, "/" marking, bend left=35] & & & & \text{rewb} \arrow[luu, "/" marking, mydeepred, dashed, swap]
			\arrow[luu, phantom, shift right=6, "\substack{\text{Theorem} \\ \text{\ref{thm:antisl}}}", pos=0.7, swap, mydeepred]
			\arrow[ruu, mydeepred, dashed, swap, "\substack{\text{Corollary} \\ \text{\ref{cor:rewbisig}}}", near end, twoheadrightarrow] \arrow[rruu, swap, "\text{\cite{berglund2023re,campeanu2003formal}}", twoheadrightarrow] \arrow[from=llluu, "/" marking, crossing over, bend left=4] \arrow[from=llluu, phantom, shift left=3, "\text{\scriptsize \cite{berglund2023re}}\,", bend left=4, pos=0.6] \\
			& & \substack{\text{Larsen's hierarchy~\cite{larsen1998regular}}} \arrow[ruu, crossing over, mydeepred, dashed, swap, "\text{Theorem~\ref{thm:hier}}", pos=0.8, twoheadrightarrow] \arrow[rrru, twoheadrightarrow, bend right=8]
		\end{tikzcd}
		\caption{The inclusion relations between the classes}
		\label{fig:langrels}
	\end{figure}

In this paper, we have shown the following five results: (1) that every rewb describes an indexed language (Corollary~\ref{cor:rewbisig}), (2) in particular that every rewb without a captured reference describes a nonerasing stack language (Corollary~\ref{cor:refcap}), (3) however that there exists a rewb that describes a non-stack language (Theorem~\ref{thm:antisl}), (4) therefore that there exists a rewb that needs a captured reference (Corollary~\ref{cor:needrefcap}), and (5) finally that Larsen's hierarchy $\syuugou[L(x_i)]{i \in \mynat}$ given in \cite{larsen1998regular} is within the class of nonerasing stack languages (Theorem~\ref{thm:hier}).
We have obtained the results by using three automata models, namely NESA, SA, and NSA, and using the semantics of rewbs given in \cite{schmid2016characterising, freydenberger2019deterministic} that treats a rewb as a regular expression allowing us to obtain the underlying NFA. 
Figure~\ref{fig:langrels} depicts the inclusion relations between the classes mentioned in the paper. Here, $A \to B$ stands for $A \subseteq B$, $A \twoheadrightarrow B$ for $A \subsetneq B$, and $A \nrightarrow B$ for $A \nsubseteq B$, respectively. A label on an arrow refers to the evidence.  A red dashed arrow indicates a novel result proved in this paper, where for a strict inclusion, we show for the first time the inclusion itself in addition to the fact that it is strict.

As future work, we would like to investigate the use of the pumping lemma for rewbs without a captured reference that can be derived from the contraposition of our Corollary~\ref{cor:refcap} and a pumping lemma for NESA~\cite{ogden1969intercalation}.  We expect it to be a useful tool for discerning which rewbs need captured references. Additionally, we suspect that our construction of NESA in Theorem~\ref{thm:hier} is useful for not just $x_i$ of \cite{larsen1998regular} but also for more general rewbs that have only one $\bs i$ for each $(_i\,)_i$, and we would like to investigate further uses of the construction.
 


\bibliography{refs}

\appendix

\section{Formal definitions and proofs for Section~\ref{sec:prelim}}
\label{app:prelim}
\begin{definition2} \label{dfn:dereference} 
	Suppose that $\bot \notin \Sigma \mydisjointu B_k \mydisjointu \slice{k}$, and we define a Turing machine with one tape $\deref_k$ as follows:
	\begin{quote}
		$\deref_k$ = ``On input string $v \in (\Sigma \mydisjointu B_k \mydisjointu \slice{k}) ^ \ast$,
				\begin{bracketenumerate}
					\item Move the head to the right until it reads an $i \in \slice{k}$, then go to \mylipicsenumitem{2}. If such $i$ does not exist, go to \mylipicsenumitem{6}.
                    \item The assertion $P_2\triangleq$ `The symbol the head points out now is the leftmost natural number $i \in \slice{k}$ on the tape' holds. Move the head to the left until it reads a $\lbrack_i$, then go to \mylipicsenumitem{3}. If such $\lbrack_i$ does not exist, go to \mylipicsenumitem{5}.
                    \item Let $P_3$ be `There is a right bracket $\rbrack_i$ between the current head position and this $i$'. If $P_3$ holds, go to \mylipicsenumitem{4}. Otherwise, go to \mylipicsenumitem{7}.
					\item Move the head to the right one by one, seeking a bracket $\rbrack_i$. Note that no $j \in \slice{k}$ can appear in this scan since $P_2$ and $P_3$ holds. Now, if the symbol written on the tape cell that the head points to is $a \in \Sigma$, then insert $a$ into the position immediately preceding this $i$, and go right; if it is $b \in B_k \backslash \syuugou{\rbrack_i}$, simply go right; if it is $\rbrack_i$, go to \mylipicsenumitem{5}.
					\item Go back to \mylipicsenumitem{1} and remove this $i$.
                    \item $P_6\triangleq$ `No $j \in \slice{k}$ is written on the tape' holds. Again scan the tape from the beginning, and remove all brackets.
					\item Erase all symbols written on the tape, and write the symbol $\bot$.''
				\end{bracketenumerate}
	\end{quote}
	\end{definition2}

    Henceforth, we refer to these step numbers as encircled numbers \maru{1}, \maru{2}, etc. The order of execution of $\deref_k$ seems to be either $(\maru{1}\maru{2}(\maru{3}\maru{4}\maru{5} + \maru{5})) ^ \ast \maru{1} (\maru{6} + \maru{2}\maru{3}\maru{7})$ or $(\maru{1}\maru{2}(\maru{3}\maru{4}\maru{5} + \maru{5}))^\omega$, however the latter cannot happen; because each time $\deref_k$ executes the loop unit $\maru{1}\maru{2}(\maru{3}\maru{4}\maru{5}+\maru{5})$, the leftmost natural number $i \in \slice{k}$ is replaced with a string that does not contain natural numbers, therefore the loop runs for at most the number of elements of $\slice{k}$ in an input string, which is bounded above by the length of the input string. Hence $\deref_k$ halts for any input.
	Thus, we think of $\deref_k$ as a computable function $(\Sigma \mydisjointu B_k \mydisjointu \slice{k}) ^ \ast \rightarrow \Sigma ^ \ast \mydisjointu \syuugou{\bot}$.

	\begin{lemma} \label{lem:derefloop}
        Suppose that the loop unit (namely \maru{1}\maru{2}\maru{3}\maru{4}\maru{5} or \maru{1}\maru{2}\maru{5}) is executed exactly $m^\prime$ times when an input string $v = v_0 n_1 v_1 \cdots n_m v_m$\,($m = \cnt{v}$) is given to $\deref_k$. Then, for each $r \in \syuugou{0,1,\dots,m'}$, let $v^{(r)}$ be the string written on the tape immediately after the $r$\textsuperscript{th} loop, and let $v_{[r]}$ be the string over $\Sigma \mydisjointu B_k$ defined as follows:
	\[
		v_{[r]} \triangleq \begin{cases}
			\text{the string bracketed in }\lbrack_i\text{ and }\rbrack_i\text{ scanned at }\maru{4}, & (\text{if }\maru{1}\maru{2}\maru{3}\maru{4}\maru{5}\text{ is executed}) \\
 			\varepsilon. & (\text{if }\maru{1}\maru{2}\maru{5}\text{ is executed})
 		\end{cases}
	\]
Under the assumptions above, for each $r \in \syuugou{0,1,\dots,m'}$ the following equality holds:
	\[
		v^{(r)} = v_0\,g(v_{[1]})\,v_1 \cdots g(v_{[r]})\,v_r n_{r+1} v_{r+1} \cdots n_m v_m.
	\]
	\end{lemma}
	
    \begin{proof} When $r = 0$, the left side is $v^{(0)}$ and the right side is $v_0 n_1 v_1 \cdots n_m v_m = v$, as required. Recall that immediately before the ($r+1$)\textsuperscript{st} loop, the string written on the tape is $v^{(r)}$. It continues as follows:
			\begin{itemize}
                \item At \maru{1}, the head of $\deref_k$ is placed at immediately after $g(v_{[r]})$ in
                \[
                	v^{(r)} = v_0\,g(v_{[1]})\,v_1 \cdots g(v_{[r]})\,v_r \uuline{n_{r+1}} v_{r+1} \cdots n_m v_m,
                \]
                and moves to $n_{r+1}$. Henceforth, we write simply $s_r$ for $v_0\,g(v_{[1]})\,v_1 \cdots g(v_{[r]})\,v_r$, which appears before $n_{r+1}$.
                \item At \maru{2}, there are two cases:
					\begin{itemize}
                        \item When $\deref_k$ follows $\maru{3}\maru{4}\maru{5}$, since it moves from \maru{2} to \maru{3} and $P_3$ holds, $s_r$ is in the form $x_0 [_{n_{r+1}} v_{[{r+1}]} ]_{n_{r+1}} x_1$. At \maru{4} and \maru{5}, the strings written on the tape after each step are
							\[
								\maru{4}\ s_r\,g(v_{[r+1]})\, \uuline{n_{r+1}}v_{r+1} \cdots n_m v_m, \text{ and } \maru{5}\ s_r\,g(v_{[r+1]})\,v_{r+1} \cdots n_m v_m.
							\]
							By definition, this string after \maru{5} is $v^{(r+1)}$.
						\item When $\deref_k$ follows \maru{5}, since $v_{[r+1]} = \varepsilon$, the string immediately after the execution of \maru{5}, namely $v^{(r+1)}$, is $s_r v_{r+1} \cdots n_m v_m = s_r\,g(v_{[r+1]})\, v_{r+1} \cdots n_m v_m.$
					\end{itemize}
					In both cases, the equation holds for $r+1$.
			\end{itemize}
			This completes the proof.
	\end{proof}
	
	\begin{lemma} \label{lem:cntloop} For a matching string $v$, the loop of $\deref_k$ runs exactly $m = \cnt{v}$ times. Hence $\deref_k(v) \in \Sigma ^ \ast$.
	\end{lemma}
    \begin{proof} Let $m ^ \prime$ be the number of loop iterations. Because obviously $m ^ \prime \leq m$, we suppose $m ^ \prime < m$ for contradiction. If so, since it is the case that the execution moves from $\maru{1}\maru{2}\maru{3}$ to \maru{7} at the ($m'+1$)\textsuperscript{st} loop, $P_3$ does not hold for $v ^ {(m^\prime)}$. By Lemma~\ref{lem:derefloop},
		\[
			v^{(m^\prime)} = \underbrace{v_0\,g(v_{[1]})\,v_1 \cdots g(v_{[m^\prime]})\,v_{m^\prime}}_{s_{m^\prime}} \uuline{n_{m^\prime+1}} v_{m^\prime+1} \cdots n_m v_m
		\]
        follows. Since $i$ in the ($m'+1$)\textsuperscript{st} loop is this $n_{m^\prime +1}$ and $s_{m^\prime} \ni \lbrack_{n_{m^\prime+1}}$, there is $v_j$ among $v_0, \dots, v_{m^\prime}$ such that $v_j \ni \lbrack_{n_{m^\prime+1}}$. This $v_j$ can be written in the form $u_0 \lbrack_{n_{m^\prime+1}} u_1$. Because $v$ is matching and $y_{m^\prime} \supseteq v_j \ni \lbrack_{n_{m^\prime+1}}$, one of $u_1, v_{j+1}, \dots, v_{m^\prime}$ contains $\rbrack_{n_{m^\prime +1}}$. This contradicts the fact that $P_3$ does not hold at \maru{3}.
	\end{proof}
	
    \begin{proof}[Proof of Lemma~\ref{lem:deref}.] In Lemma~\ref{lem:cntloop}, the string $v ^ {(m)}$, which is written on the tape immediately after the $m$\textsuperscript{th} loop, becomes $g(v ^ {(m)})$ at \maru{6}, and therefore $\deref_k$ halts, as required.	
	\end{proof}

    \begin{proof}[Proof of Lemma~\ref{lem:prefix}]
      Immediate from Lemma~\ref{lem:deref}.
    \end{proof}

    Finally, we prove that every $v \in \refwords_k$ (see Definition~\ref{dfn:reflang}) is matching, concluding $L(\alpha)\subseteq \Sigma^\ast$ with Lemma~\ref{lem:cntloop}.
    \begin{lemma} \label{lem:refwsup} The following facts hold where $\alpha \in \rewb_k$ and $i \in \slice{k}$:
		\begin{alphaenumerate}
			\item $\reflang_k(\alpha) \subseteq (\Sigma \mydisjointu \syuugou[\lbrack_j, \rbrack_j]{j \in \varf{(\alpha)}} \mydisjointu \varf{(\alpha)}) ^ \ast$, 
			\item $\forall v_1, v_2.\, v_1 [_i v_2 \in \reflang_k(\alpha)$ $\Longrightarrow$ $\exists v_2^\prime, v_3.\, v_2 = v_2 ^ \prime \rbrack_i v_3\,\land\, v_2^\prime \niton \lbrack_i, \rbrack_i, i$. 
		\end{alphaenumerate}
	\end{lemma}
	\begin{proof}
	Immediate from the definitions of $\reflang_k(\alpha)$ and $\varf{(\alpha)}$.
	\end{proof}
	
    \begin{proof}[Proof of Lemma~\ref{lem:refw}] There exists $\alpha \in \rewb_k$ such that $v \in \reflang_k(\alpha)$. Hereafter, we let $m = \cnt{v}$ and $v = v_0 n_1 v_1 \cdots n_m v_m$. For each $r \in \syuugou{1, \dots, m}$, if $y _ {r-1}$ can be written in the form $x_1 \lbrack_{n_r} x_2$, we obtain $v = x_1 \lbrack_{n_r} x_2 n_r v_r \cdots n_m v_m$ and by Lemma~\ref{lem:refwsup} (b), $x_2 n_r v_r \cdots n_m v_m$ can be written in the form $x_2^\prime \rbrack_{n_r} x_3$ with $x_2^\prime \niton \lbrack_{n_r}, \rbrack_{n_r}, n_r$. Here $x_2 ^ \prime \rbrack_{n_r} \niton n_r$ holds. Therefore, $x_2^\prime \rbrack_{n_r}$ is a prefix of $x_2$ and of course $x_2^\prime \niton \lbrack_{n_r}, \rbrack_{n_r}$. Hence, $v$ is matching.
	\end{proof}
 \section{Proof of Corollary~\ref{cor:extend}}
\label{app:extend}
\begin{proof}[Proof of Corollary~\ref{cor:extend}.]
    There are a string $w^\prime \in (\Sigma \mydisjointu B_k \mydisjointu \slice{k}) ^ \ast$ and a final state $q_f \in F$ such that $q_0 \mytto{w}{N} q \mytto{w^\prime}{N} q_f$. Hence, $w w^\prime \in L(N) = \reflang_k(\alpha)$ follows. By Lemma~\ref{lem:refw} the string $ww^\prime$ is matching, therefore by Lemma~\ref{lem:prefix} its prefix $w$ is matching.
\end{proof}
 \section{Proof of Lemma~\ref{lem:3}}
\label{app:3}
\fullversion{First, we define some notations (see also Appendix~\ref{app:prelim} for the formal definitions of $v_{[r]}$ and $v^{(r)}$).}{First, we recall the notation $v_{[r]}$ (cf. Section~\ref{sec:prelim}) and explain a new notation $v^{(r)}$ informally (see the full version~\cite{nogami2023expressive} for the formal definitions of $v_{[r]}$ and $v^{(r)}$). Let $k$ be a positive integer and $v = v_0 n_1 v_1 \dots n_m v_m$ ($m = \cnt{v}$) a matching string over $\Sigma \mydisjointu B_k \mydisjointu \slice{k}$. For each $r = 1, 2, \dots$, the notation $v_{[r]}$ denotes the string which $\deref_k$ scans at the $r$\textsuperscript{th} number $n_r$ and $v^{(r)}$ the string immediately after the $r$\textsuperscript{th} replacement (also we let $v^{(0)} = v$).
    For example, in the case of $v = \lbrack_1 a\,\lbrack_2 b \rbrack_2\, 2\,\rbrack_1\,1$, $\deref_k$ processes $v$ as follows, therefore $v_{[1]} = b$ and $v_{[2]} = a \lbrack_2 b \rbrack_2 b$: $v^{(0)} = \lbrack_1 a\,\underline{\lbrack_2 b \rbrack_2}\, 2\,\rbrack_1\,1 \to v^{(1)} = \underline{\lbrack_1 a\,\lbrack_2 b \rbrack_2\, b\,\rbrack_1}\,1 \to v^{(2)} = \lbrack_1 a\,\lbrack_2 b \rbrack_2\, b\,\rbrack_1\,abb \to abbabb$. We can easily prove the following claim (see also the full version~\cite{nogami2023expressive}): For each $r \in \syuugou{0,1,\dots,m}$, 
    \begin{equation}
        v^{(r)} = v_0\,g(v_{[1]})\,v_1 \cdots g(v_{[r]})\,v_r n_{r+1} v_{r+1} \cdots n_m v_m. \tag{$\ast$} \label{claim:derefloop}
    \end{equation}

    We prepare some more notations for the proof.} Let $\Sigma_{\bot}^\ast \triangleq \Sigma^\ast \mydisjointu \syuugou{\bot}$ and for every $w \in \Sigma_{\bot}^\ast$ and $s \in \Sigma^\ast$, let $w/s$ denote the string $w$ but with the suffix $s$ erased if $w$ ends with $s$, and otherwise $\bot$.
	To use this notation, we expand the set of all IDs $I = Q \times \Sigma^\ast \times \syuugou{\#}(\Gamma \mydisjointu \syuugou{\mycent, \mydollar, \,\myspl}) ^ \ast \syuugou{\mydollar}$ to $I_\bot \triangleq Q \times \Sigma_{\bot}^\ast \times \syuugou{\#} (\Gamma \mydisjointu \syuugou{\mycent,\mydollar,\myspl}) ^ \ast \syuugou{\mydollar}$, and we let $C(w)$ denote an ID $C = (\cdot, w, \cdots)$ and $\myid{(n)}^\prime$ denote the following binary relation over $I_\bot$:	
	\[
		\myid{(n)}^\prime\,\triangleq \syuugou[(C(w), C^\prime(w/(a_i \cdots a_{i+\overline{d}-1})))]{C(a_i \cdots a_k) \vdash_{(n)} C^\prime(a_{i+\overline{d}} \cdots a_k), w \in \Sigma_{\bot}^\ast}.
	\]
	In addition, we define $\myid{}^\prime\,\triangleq \bigcup_{n}{\,\myid{(n)}^\prime}$.	
	Then, $\myid{}\,\subseteq\,\myid{}^\prime$ is immediate and we show that the converse partially holds, in the sense that:
	
	\begin{lemma} \label{lem:idprime2id} For every string $w \in \Sigma^\ast$ and $w^\prime \in \Sigma^\ast$, $C(w) \myid{}^\prime C^\prime(w^\prime)$ implies $C(w) \myid{} C^\prime(w^\prime)$.
	\end{lemma}
	
    \begin{proof} By the definition of $\mathord{\myid{}'}$, there is $(C(a_i \cdots a_k), C^\prime(a_{i+\overline{d}} \cdots a_k)) \in \mathord{\myid{}}$ such that $w^\prime = w / (a_i \cdots a_{i+\overline{d}-1})$. By the definition of $\mathord{\myid{}}$, $C(w) = C(a_i \cdots a_{i+\overline{d}-1} w^\prime) \myid{} C^\prime(w^\prime)$ holds.
	\end{proof}
    \begin{definition2} Given $C, C^\prime \in I_\bot$, we write $C \myonlymove{(n)} C^\prime$ if $C \myid{(n)}^\prime C^\prime$ and $\forall j,C''. C \myid{(j)}' C'' \Longrightarrow j = n  \,\land\, C'' = C'$. We often omit the subscript $(n)$ and simply write $C \myonlymove{} C^\prime$. Note that $C \myonlymove{} C^\prime$ implies not only $C \myid{}^\prime C^\prime$ but also determinism:
    $\forall C^{\prime\prime} \in I_\bot .\, C \myid{}^\prime C^{\prime\prime} \Longrightarrow C^\prime = C^{\prime\prime}$.
	\end{definition2}

	\begin{lemma} \label{lem:3sup}
	Suppose that $\gamma \in (\Sigma \mydisjointu B_k \mydisjointu \slice{k}) ^ \ast$, $i \in \slice{k}$, $w \in \Sigma^\ast$, $\beta \in (\Gamma \mydisjointu \syuugou{\mycent,\mydollar})^\ast$ and $p \in Q_N \mydisjointu \syuugou[\symexec_i]{i \in \slice{k}}$. Let $m = \cnt{(\gamma i)}\,(\geq 1)$. If $\gamma i$ is matching, 
	\[
		(\symcall_i, w, \# Z_0 \gamma \mycent p \myspl \! \mydollar i \beta \mydollar) \myonlymove{} \cdots \myonlymove{} (\symret_i, w /g((\gamma i)_{[m]}), \# Z_0 \gamma \mycent p \myspl\! \mydollar i \beta \mydollar)
	\]
        holds, where no ID with a state in $Q_N$ appears in the calculation $\cdots$.
	\end{lemma}
	 
	\begin{proof} In this proof, we sometimes write the stack representation $\# \cdots Z \myspl\!\cdots\mydollar$ as $\# \cdots \myspr\!Z \cdots\mydollar$ with the head-reversed arrow $\mathord{\myspr}$. First, if $\gamma \niton \lbrack_i$, it holds that
		\begin{align*}
			(\symcall_i, w, \# Z_0 \gamma \mycent p \myspl\! \mydollar i \beta \mydollar) &\myonlymove{} ^ \ast (\symcall_i, w, \# Z_0 \myspl\! \gamma \mycent p \mydollar i \beta\mydollar) \\
			&\myonlymove{}\phantom{^\ast} (\symret_i, w, \#Z_0\!\myspr \gamma \mycent p \mydollar i \beta\mydollar) \myonlymove{} ^ \ast (\symret_i, w, \# Z_0 \gamma \mycent p \myspl\! \mydollar i \beta \mydollar),
		\end{align*}
		and by $(\gamma i) _ {[m]} = \varepsilon$, we have $w = w /g((\gamma i)_{[m]})$, as required.
		Henceforth, we assume that $\gamma \ni \lbrack_i$ and the decomposition $\gamma = \gamma_0 \lbrack_i \gamma_1$ ($\gamma_1 \niton \lbrack_i$). Moreover, we can further decompose $\gamma _ 1 = \gamma _ 2 \rbrack_i \gamma_3$ ($\gamma_2 \niton \lbrack_i, \rbrack_i$, $\gamma_3 \niton \lbrack_i$) because $\gamma i$ is matching. We prove by induction on $m$.

        Case $m = 1$: By $\cnt{\gamma} = 0$, $\gamma _ 2 \in (\Sigma \mydisjointu B_k) ^ \ast$ follows. Letting $w^\prime \triangleq w / g(\gamma_2)$, we have
				\begin{align*}
					& (\symcall_i,w, \#Z_0 \gamma_0 \lbrack_i \gamma_1 \mycent p \myspl\! \mydollar i \beta\mydollar) \myonlymove{}^\ast (\symcall_i, w, \# Z_0 \gamma_0 \lbrack_i\,\, \myspl\! \gamma_1 \mycent p \mydollar i \beta\mydollar) \myonlymove{} (\symexec_i, w, \# Z_0 \gamma_0 \lbrack_i \myspr\!\gamma_1 \mycent p \mydollar i \beta\mydollar) \\
					& \myonlymove{}^{\ast} (\symexec_i, w^\prime, \# Z_0 \gamma_0 \lbrack_i \gamma_2 \rbrack_i \myspl\! \gamma_3 \mycent p \mydollar i \beta\mydollar) \myonlymove{} (\symret_i, w^\prime, \# Z_0 \gamma_0 \lbrack_i \gamma_2 \rbrack_i \!\myspr\!\gamma_3 \mycent p \mydollar i \beta\mydollar)\\
					& \myonlymove{}^{\ast} (\symret_i, w^\prime, \# Z_0 \gamma_0 \lbrack_i \gamma_2 \rbrack_i \gamma_3 \mycent p \myspl\! \mydollar i \beta\mydollar).
				\end{align*}
			Therefore, the claim holds since no ID with a state in $Q_N$ appears in this calculation and $\gamma_2 = (\gamma i)_{[m]}$ follows from $(\gamma i)^{(0)} = \gamma i = \gamma_0 \lbrack_i \gamma_2 \rbrack_i \gamma_3 i$, $\gamma_3 \niton \lbrack_i$.
	
        Case $\syuugou{1,\dots,m} \Longrightarrow m+1$: Let $m_0 \triangleq \cnt{\gamma_0}$ and $l \triangleq \cnt{\gamma_2}\,(\geq 0)$. Now, $m_0 + l \leq m = \cnt{\gamma}$ holds and we write $\gamma_2 = \lambda_0 n_{m_0+1} \lambda_1 \cdots n_{m_0+l} \lambda_l$. We also define $\eta _ {r} \triangleq \gamma_0 \lbrack_i \lambda_0 n_{m_0+1} \cdots \lambda_{r-1} n_{m_0+r}$ for each $r \in \syuugou{1, \dots, l}$. By $\eta _ {r}$ being a prefix of $\gamma i$ and Lemma~\ref{lem:prefix}, $\eta_r$ is matching and $(\eta_r)_{[m_0 + r]} = (\gamma i) _ {[m_0 + r]}, r \in\syuugou{1,\dots,l}$ holds.
				In particular, it follows that $n_{m_0+r} \neq i$ for every $r$ (if there is $r$ such that $n_{m_0+r} = i$, $\gamma_2 \supseteq \lambda_0 n_{m_0+1} \cdots \lambda_{r-1} \ni \mathord{\rbrack_i}$ holds but this contradicts $\gamma_2 \niton \mathord{\rbrack_i}$). Thus, letting $w_0 \triangleq w$, $w_r^\prime \triangleq w_{r-1} / g(\lambda_{r-1})$, $w_r \triangleq w_r^\prime / g((\eta_r)_{[m_0+r]})$ and $w^\prime = w_l / g(\lambda_{l})$, we have
				\begin{align*}
					& (\symcall_i, w, \# Z_0 \gamma \mycent p \myspl\! \mydollar i \beta \mydollar) \\
					&\myonlymove{} ^\ast (\symcall_i, w, \# Z_0 \gamma_0 \lbrack_i\,\, \myspl\! \lambda_0 n_{m_0+1} \lambda_1 \cdots n_{m_0+l} \lambda_l \rbrack_i \gamma_3 \mycent p \mydollar i \beta \mydollar) \\
					& \myonlymove{}\phantom{^\ast} (\symexec_i, w_0, \# Z_0 \gamma_0 \lbrack_i \myspr\! \lambda_0 n_{m_0+1} \lambda_1 \cdots n_{m_0+l} \lambda_l \rbrack_i \gamma_3 \mycent p \mydollar i \beta \mydollar) \\
					& \myonlymove{}^\ast (\symexec_i, w_1^\prime, \# Z_0 \gamma_0 \lbrack_i \lambda_0 n_{m_0+1} \myspl \lambda_1 \cdots n_{m_0+l} \lambda_l \rbrack_i \gamma_3 \mycent p \mydollar i \beta \mydollar) \\
					& \myonlymove{}\phantom{^\ast} (\symcall_{n_{m_0+1}}, w_1^\prime, \# Z_0 \gamma_0 \lbrack_i \lambda_0 \mycent r_i \myspl\! \mydollar n_{m_0+1} \lambda_1 \cdots n_{m_0+l} \lambda_l \rbrack_i \gamma_3 \mycent p \mydollar i \beta \mydollar) \\
					& \myonlymove{} ^ \ast (\symret_{n_{m_0+1}}, w_1, \# Z_0 \gamma_0 \lbrack_i \lambda_0 \mycent r_i \myspl\! \mydollar n_{m_0+1} \lambda_1 \cdots n_{m_0+l} \lambda_l \rbrack_i \gamma_3 \mycent p \mydollar i \beta \mydollar) \\
                    & \hspace{1.5in} \text{(by $\eta_1$ being matching and induction hypothesis)} \\
					& \myonlymove{}\phantom{^\ast} (E_{\symexec_i,n_{m_0+1}}, w_1, \# Z_0 \gamma_0 \lbrack_i \lambda_0 \mycent \myspl\! \mydollar n_{m_0+1} \lambda_1 \cdots n_{m_0+l} \lambda_l \rbrack_i \gamma_3 \mycent p \mydollar i \beta \mydollar) \\
					& \myonlymove{}\phantom{^\ast} (L_{\symexec_i,n_{m_0+1}}, w_1, \# Z_0 \gamma_0 \lbrack_i \lambda_0 n_{m_0+1} \myspl\! \lambda_1 \cdots n_{m_0+l} \lambda_l \rbrack_i \gamma_3 \mycent p \mydollar i \beta \mydollar) \\
					& \myonlymove{}\phantom{^\ast} (\symexec_i, w_1, \# Z_0 \gamma_0 \lbrack_i \lambda_0 n_{m_0+1} \!\myspr\! \lambda_1 \cdots n_{m_0+l} \lambda_l \rbrack_i \gamma_3 \mycent p \mydollar i \beta \mydollar) \\
					\myonlymove{}^\ast \cdots & \myonlymove{}^\ast (\symexec_i, w_{l}, \# Z_0 \gamma_0 \lbrack_i \lambda_0 n_{m_0+1} \lambda_1 \cdots n_{m_0+l} \! \myspr \!\lambda_l \rbrack_i \gamma_3 \mycent p \mydollar i \beta \mydollar) \\
                    & \hspace{1.5in} \text{(by similar calculation and induction hypothesis)} \\
					& \myonlymove{}^\ast (\symexec_i, w^\prime, \# Z_0 \gamma_0 \lbrack_i \lambda_0 n_{m_0+1} \lambda_1 \cdots n_{m_0+l} \lambda_l \rbrack_i \myspl\!\gamma_3 \mycent p \mydollar i \beta \mydollar) \\
					& \myonlymove{}\phantom{^\ast} (\symret_i, w^\prime, \# Z_0 \gamma_0 \lbrack_i \lambda_0 n_{m_0+1} \lambda_1 \cdots n_{m_0+l} \lambda_l \rbrack_i \!\myspr\! \gamma_3 \mycent p \mydollar i \beta \mydollar) \\
					& \myonlymove{}^\ast (\symret_i, w^\prime, \# Z_0 \gamma_0 \lbrack_i \lambda_0 n_{m_0+1} \lambda_1 \cdots n_{m_0+l} \lambda_l \rbrack_i \gamma_3 \mycent p \myspl\!\mydollar i \beta \mydollar),
				\end{align*}
			and
			\begin{align*}
				w^\prime &= w/g(\lambda_0)\,g((\eta_1)_{[m_0+1]})\,g(\lambda_1) \cdots g((\eta_l)_{[m_0+l]})\,g(\lambda_l) \\
				&= w / g(\lambda_0)\,g((\gamma i)_{[m_0+1]})\,g(\lambda_1) \cdots g((\gamma i)_{[m_0+l]})\,g(\lambda_l),
			\end{align*}
            where no ID with a state in $Q_N$ appears in this calculation. Here, we write
            \[
            	\gamma = \gamma_0 \lbrack_i \lambda_0 n_{m_0+1} \lambda_1 \cdots n_{m_0+l} \lambda_l \rbrack_i \gamma_3 = v_0 n_1 v_1 \cdots n_{m} v_{m}
            \]
            and decompose its substrings as 			
			\[
				v_{m_0} = \chi_0 \lbrack_i \lambda_0,\quad v_{m_0+l} = \lambda_l \rbrack_i \chi_1,\ \text{and}\ \gamma_3 = \chi_1 n_{m_0+l+1} v_{m_0+l+1} \cdots n_{m} v_{m}.
			\]
            Then, by \fullversion{Lemma~\ref{lem:derefloop}}{equation~\eqref{claim:derefloop}}, we can write $(\gamma i) ^ {(m)}$ as
{\small
\[
	v_0 \cdots \rlap{$\underbrace{\phantom{\chi_0 \lbrack_i \lambda_0 }}_{v_{m_0}}$} \chi_0 \lbrack_i \rlap{$\overbrace{\phantom{\lambda_0 \,g((\gamma i)_{[m_0+1]})\, v_{m_0+1} \cdots g((\gamma i)_{[m_0+l]})\, \lambda_l}}^{= (\gamma i) _ {[m+1]}}$} \lambda_0 \,g((\gamma i)_{[m_0+1]})\, v_{m_0+1} \cdots g((\gamma i)_{[m_0+l]})\, \rlap{$\underbrace{\phantom{\lambda_l \rbrack_i\chi_l}}_{v_{m_0+l}}$} \lambda_l \rbrack_i \overbrace{\chi_1\,g((\gamma i)_{[m_0+l+1]}) v_{m_0+l+1}  \cdots g((\gamma i)_{[m]})\, v_{m}}^{\gamma_3^\prime \triangleq} i.
\]
}
			That is, it holds that $\gamma_3 ^ \prime \triangleq \chi_1\,g((\gamma i)_{[m_0+l+1]}) v_{m_0+l+1} \cdots g((\gamma i)_{[m]})\, v_{m} \niton [_i$ by $\gamma_3 \niton \lbrack_i$, and we obtain $(\gamma i) _{[m+1]} = \lambda_0\,g((\gamma i)_{[m_0+1]})\,\lambda_1 \cdots\,g((\gamma i)_{[m_0+l]})\,\lambda_l$, as shown above. Therefore, the claim holds for $m+1$ since $w^\prime = w/g((\gamma i)_{[m+1]})$.
	\end{proof}
	\begin{proof}[Proof of Lemma~\ref{lem:3}.]
		For arbitrary $w \in \Sigma^\ast$, by Lemma~\ref{lem:3sup},
		\begin{align*}
 			& (q, w, \# Z_0 \gamma \myspl\! \mydollar) \myid{(3)} (W_{q^\prime}, w, \# Z_0 \gamma i \myspl\! \mydollar) \myonlymove{(4)} (\symcall_i, w, \# Z_0 \gamma \mycent q^\prime \myspl\! \mydollar i \mydollar) \\
 			& \myonlymove{} ^\ast (\symret_i, w / g((\gamma i)_{[m]}), \# Z_0 \gamma \mycent q^\prime \myspl\! \mydollar i \mydollar)\myonlymove{(14)} (E_{q^\prime,i}, w / g((\gamma i)_{[m]}), \# Z_0 \gamma \mycent \myspl\! \mydollar i \mydollar) \\
            &\myonlymove{(15)} (L_{q^\prime,i}, w / g((\gamma i)_{[m]}), \# Z_0 \gamma i \myspl\! \mydollar) \myonlymove{(17)} (q^\prime, w / g((\gamma i)_{[m]}), \# Z_0 \gamma i \myspl\! \mydollar) \tag{\fullversion{$\ast$}{$\ast\ast$}} \label{keisan}
 		\end{align*}
 		holds. Assuming (a), we can replace $\mathord{\myonlymove{}}$ in equation~\eqref{keisan} with $\mathord{\myid{}}$ by Lemma~\ref{lem:idprime2id} because $w / g((\gamma i)_{[m]}) = w^\prime \in \Sigma^\ast$ holds, and therefore, (b) follows.
        Supposing (b) conversely, we have $(q, w, \# Z_0 \gamma \myspl\! \mydollar) \myid{(3)} (W_{q^\prime}, w, \# Z_0 \gamma i \myspl\! \mydollar) \myid{}^\prime \cdots \myid{}^\prime (p, w^\prime, \# \beta \mydollar)$, where no ID with a state in $Q_N$ appears in either this calculation or \eqref{keisan} except in their leftmost and rightmost IDs. Therefore, their two calculations coincide by the determinism of $\mathord{\myonlymove{}}$. In particular, we obtain $p = q^\prime$, $w^\prime = w/ g((\gamma i)_{[m]})$ and $\beta = Z_0 \gamma i \myspl\!$ by the equality of their rightmost IDs, and thus, (a) follows because $w^\prime \in \Sigma^\ast$.
	\end{proof}
 \section{Proof of Theorem~\ref{thm:main}}
\label{app:main}
\begin{proof}[Proof of Theorem~\ref{thm:main}.]
    Taking any $w \in L(\alpha)$, we have $v = v_0 n_1 v_1 \cdots n_m v_m \in \reflang_k(\alpha) = L(N)$ (here, $m = \cnt{v}$) such that $w = \deref_k(v)$. Now, suppose that $q_0 \underset{N}{\overset{v_0}{\Longrightarrow}} q_{(0)} \underset{N}{\overset{n_1 v_1}{\Longrightarrow}} q_{(1)} \underset{N}{\overset{n_2 v_2}{\Longrightarrow}} \cdots \underset{N}{\overset{n_m v_m}{\Longrightarrow}} q_{(m)} \in F$. Letting $y_r \triangleq v_0 n_1 v_1 \cdots n_r v_r$ and $z_r \triangleq g(v_0)\,g(v_{[1]})\,g(v_1) \cdots g(v_{[r]})\,g(v_r)$ for each $r \in \syuugou{0,1,\dots,m}$, we show by induction on $r$ that for any $w^\prime \in \Sigma ^ \ast$, $(q_0, z_rw ^ \prime, \# Z_0 \myspl\! \mydollar) \myid{} ^ \ast (q_{(r)}, w^\prime, \# Z_0 y_r \myspl\! \mydollar)$ holds.
	
            Case $r = 0$: It holds that $q_0 \underset{N}{\overset{v_0}{\Longrightarrow}} q_{(0)}$ and $y_0 = v_0 \in (\Sigma \mydisjointu B_k) ^ \ast$. Letting $\gamma = \varepsilon$ in Lemma~\ref{lem:12}, we obtain $(q_0, g(v_0)w ^ \prime, \# Z_0 \myspl\! \mydollar) \myid{} ^ \ast (q_{(0)}, w^\prime, \# Z_0 y_0 \myspl\! \mydollar)$, as required.

            Case $r \Rightarrow r+1$: By the induction hypothesis, it holds that $(q_0, z_{r+1}w ^ \prime, \# Z_0 \myspl\! \mydollar) \myid{} ^ \ast (q_{(r)}, g(v_{[r+1]})\,g(v_{r+1})\,w^\prime, \# Z_0 y_r \myspl\! \mydollar)$. By Lemma~\ref{lem:refw}, $v$ is matching. Hence, $v$'s prefix $y_{r}n_{r+1}$ is also matching by Lemma~\ref{lem:prefix}. Let $q_{(r)} \underset{N}{\overset{n_{r+1}v_{r+1}}{\Longrightarrow}} q_{(r+1)}$ be $q_{(r)} \underset{N}{\overset{n_{r+1}}{\longrightarrow}} q^\prime \underset{N}{\overset{v_{r+1}}{\Longrightarrow}} q_{(r+1)}$. Because $\cnt{(y_{r}n_{r+1})} = r+1$, the following calculation holds by Lemma~\ref{lem:3} ${\rm (a)} \Longrightarrow {\rm (b)}$:
			\begin{align*}
				& (q_{(r)}, g((y_r n_{r+1})_{[r+1]})\,g(v_{r+1})\,w^\prime, \# Z_0 y_r \myspl\! \mydollar) \\
				& \myid{(3)} (W_{q^\prime}, g((y_r n_{r+1})_{[r+1]})\,g(v_{r+1})\,w^\prime, \# Z_0 y_r n_{r+1} \myspl\! \mydollar) \\
				& \vdash ^ \ast (q^\prime, g(v_{r+1})\,w^\prime, \# Z_0 y_r n_{r+1} \myspl\! \mydollar) \myid{} ^ \ast (q_{(r+1)}, w^\prime, \# Z_0 y_r n_{r+1} v_{r+1} \myspl\! \mydollar).
			\end{align*}
            By Lemma~\ref{lem:prefix}, it holds that $(y_r n_{r+1})_{[r+1]} = v_{[r+1]}$ and $y_r n_{r+1} v_{r+1} = y_{r+1}$, as required.
        In particular, letting $r = m$ and $w ^ \prime = \varepsilon$ in the claim, we obtain $(q_0, w, \# Z_0 \myspl\! \mydollar) \myid{} ^ \ast (q_{(m)}, \varepsilon, \#Z_0 y_m \myspl\! \mydollar)$ because $w = \deref_k(v) = z_m$ by Lemma~\ref{lem:deref}. Therefore, $w \in L(A_\alpha)$ holds since $q_{(m)} \in F$.

		Conversely, take any $w \in L(A_\alpha)$. There exist $m$ and an ID $C_{(r)} = (p_r, w_r, \# \beta _ {r} \mydollar)$ for each $r \in \syuugou{1,\dots,m}$ such that $C_{(1)} = (q_0,w, \# Z_0 \myspl\! \mydollar) \vdash \cdots \vdash C_{(r)} \vdash \cdots \vdash C_{(m)} = (p_m, \varepsilon, \# \beta_m \mydollar)$, where $p_m \in F$. We show by induction on $r$ that for each $r = 1 ,\dots, m$, the following claim holds: 
        if $p_{r} \in Q_N$, there is $\gamma_r$ such that $C_{(r)} = (p_r, w_r, \# Z_0 \gamma_{r} \myspl\! \mydollar)$, (a) $\gamma_r \in (\Sigma \mydisjointu B_k \mydisjointu \slice{k}) ^ \ast$, (b) $w = \deref_k(\gamma_r)\,w_r$ and (c) $q_0 \underset{N}{\overset{\gamma_r}{\Longrightarrow}} p_r$.
			
			Case $r=1$: It holds that $p_{1} = q_0 \in Q_N$ and $C_{(1)} = (q_0, w, \# Z_0 \myspl\! \mydollar)$. Letting $p_1 = q_0$, $w_1 = w$ and $\gamma_1 = \varepsilon$, we have (a) $\gamma_1 = \varepsilon \in (\Sigma \mydisjointu B_k \mydisjointu \slice{k}) ^ \ast$, (b) $w = \deref_k(\gamma_1)\,w_1$ and (c) $q_0 \underset{N}{\overset{\gamma_1}{\Longrightarrow}} p_1 = q_0$, as required.	
			
            Case $\syuugou{1,\dots,r} \Rightarrow r+1$: Suppose that $p_{r+1} \in Q_N$. We can define $j$ as the maximum of the set \syuugou[1 \leq j \leq r]{p_j \in Q_N} since $p_1 \in Q_N$. Rules that can be applied to $C_{(j)}$ are limited to (1), (2) and (3) because $p_j \in Q_N$. 
			
                    In the case of (1), $j = r$ holds because $p_{j+1} \in Q_N$. By $p_r \in Q_N$ and the induction hypothesis, we have $C_{(r)} = (p_r, w_r, \# Z_0 \gamma_r \myspl\! \mydollar)$, and $p_r$, $w_r$ and $\gamma_r$ satisfy conditions (a), (b) and (c). Hence, by Lemma~\ref{lem:12}\,(a)$_2$, there is $a \in \Sigma$ such that $p_{r} \myto{a}{N} p_{r+1}$, $w_{r} = a w_{r+1}$ and $\beta_{r+1} = Z_0 \gamma_r a \myspl\!$. 
					Now, we let $\gamma_{r+1} = \gamma_r a$. Since $\beta _ {r+1} = Z_0 \gamma _ {r+1}\!\myspl$\,,
					(a) $\gamma_{r+1} = \gamma_r a \in (\Sigma \mydisjointu B_k \mydisjointu \slice{k}) ^ \ast$ because $\gamma _ r \in (\Sigma \mydisjointu B_k \mydisjointu \slice{k}) ^ \ast$, (b) $\deref_k(\gamma_{r+1}) w_{r+1} = \deref_k(\gamma_r)a w_{r+1} = \deref_k(\gamma_r) w_r = w$, and (c) $q_0 \mytto{\gamma_{r+1}}{N} p_{r+1}$ because $q_0 \mytto{\gamma_r}{N} p_r \myto{a}{N} p_{r+1}$.
					The case of (2) follows similarly as the case of (1) with Lemma~\ref{lem:12}\,(b)$_2$.		
                    In the case of (3), there is $q^\prime \in Q_N$ such that $p_j \myto{i}{N} q^\prime$. By $p_j \in Q_N$ and the induction hypothesis, we have $C_{(j)} = (p_j, w_j, \# Z_0 \gamma_j \myspl\! \mydollar)$, and $p_j$, $w_j$ and $\gamma_j$ satisfy the conditions (a), (b) and (c).
                    Because $q_0 \mytto{\gamma_j}{N} p_j \myto{i}{N} q^\prime$, $\gamma_j i$ is matching by Corollary~\ref{cor:extend}. Besides, it holds that $p_{r+1} \in Q_N$ and $(p_j, w_j, \# Z_0 \gamma_j \myspl\! \mydollar) \vdash_{(3)} (W_{q^\prime}, w_j, \#Z_0 \gamma_j i \myspl\!\mydollar) \myid{} \cdots \myid{} C_{(r+1)}$, where no ID with a state in $Q_N$ appears in the calculation $\cdots$. Hence, by Lemma~\ref{lem:3} ${\rm (b)} \Longrightarrow {\rm (a)}$, we have $p_{r+1} = q^\prime$, $w_j = g((\gamma_j i)_{[m]}) w_{r+1}$, and $\beta_{r+1} = Z_0 \gamma_{j} i \myspl\!$.
					Now, we let $\gamma_{r+1} = \gamma_j i$. Since $\beta_{r+1} = Z_0 \gamma_{r+1}\myspl\!$, 
					(a) $\gamma_{r+1} = \gamma_r i \in (\Sigma \mydisjointu B_k \mydisjointu \slice{k}) ^ \ast$ because $\gamma _ r \in (\Sigma \mydisjointu B_k \mydisjointu \slice{k}) ^ \ast$, (b) $\deref_k(\gamma_{r+1}) w_{r+1} = \deref_k(\gamma_j)\,g((\gamma_j i)_{[m]})\,w_{r+1} = \deref_k(\gamma_j) w_j = w$, and (c) $q_0 \mytto{\gamma_{r+1}}{N} p_{r+1}$ because $q_0 \mytto{\gamma_j}{N} p_j \myto{i}{N} p_{r+1}$.
				
				Therefore, the claim holds for the case of $r+1$.
        In particular, letting $r = m$ in the claim, we have $C_{(m)} = (p_m, w_m, \#Z_0 \gamma_m\!\myspl\!\!\mydollar)$, and $p_m$, $w_m$ and $\gamma_m$ satisfy (a), (b) and (c) (note that $p_m \in F \subseteq Q_N$). Because $w_m = \varepsilon$, it holds that $w = \deref_k(\gamma_m)$ and that $q_0 \mytto{\gamma_m}{N} p_m \in F$, or $\gamma_m \in L(N) = \reflang_k(\alpha)$. Therefore, we have $w \in \deref_k(\reflang_k(\alpha)) = L(\alpha)$.
\end{proof}
 \section{Supplement to Corollary~\ref{cor:refcap}}
\label{app:anbnNESA}
We show that the language $T \triangleq \syuugou[a^n b^n]{n \in \mynat}$ is described by the NESA $A$ given in Figure~\ref{fig:anbnNESA}.
			\begin{figure}[htb]
				\centering
				\begin{tikzpicture}[
						shorten >=1pt,
						node distance=2.25cm,
						on grid,
						>={Stealth[round]},
						auto
					]
					\tikzset{el/.style = {inner sep=2pt, align=left, sloped}}
					\tikzset{every node/.style={font=\small}}
					\tikzset{every loop/.style={min distance=6mm,looseness=8}}
					\tikzset{every state/.style={inner sep=.1mm, minimum size=5mm}}
					\tikzset{initial text=}
					\tikzset{every initial by arrow/.style={->}}
					\tikzset{bend angle=30}
			
					\node[state,initial] (0) at (0,-3) {$q_0$};
					\node[state] (1) at (2,-3) {$q_1$};
					\node[state,accepting] (2) at (4,-3) {$q_2$};
					
					\path[->] (0) edge node {\footnotesize
									$\varepsilon/\mystop$
								} (1);
					
					\path[->] (0) edge [in=60,out=120,loop above] node {\footnotesize
									$a/ \figto{\mydollar}{\star\mydollar}$
								} ();
					\path[->] (1) edge [in=60,out=120,loop above] node {\Large
									$\substack{
										b/ \star,\,\myleft \\
										b/ \star\mydollar,\,\myleft
									}$
								} ();
					\path[->] (1) edge node {\footnotesize
									$\varepsilon/Z_0,\mystop$
								} (2);
				\end{tikzpicture}
                \caption{An NESA that recognizes $T = \syuugou[a^n b^n]{n \in \mynat}$}
                \label{fig:anbnNESA}
			\end{figure}
The set of stack symbols $\Gamma$ is $\syuugou{Z_0,\star}$, where $\star$ is a distinguished character. Intuitively, $A$ first pushes $\star$ while consuming the input $a$ to count the occurrences of $a$ and leaves $q_0$ for $q_1$ nondeterministically. At $q_1$, $A$ moves its stack pointer leftward while consuming the input $b$, and leaves there for the accepting state $q_2$ if and only if the stack pointer reaches the bottom of the stack. Finally, $A$ halts at $q_2$. Here is the proof of $L(A) = T$:
\begin{proof} 
	The inclusion $T \subseteq L(A)$ follows from the following calculation for any $n \in \mynat$:
	\begin{align*}
		(q_0,a^n b^n, \#Z_0\myspl\!\mydollar) &\myid{}^n (q_0,b^n,\#Z_0\,\star^n\myspl\!\mydollar) \\
		&\myid{}\phantom{^n} (q_1,b^n,\#Z_0\,\star^n\myspl\!\mydollar) \myid{}^n (q_1,\varepsilon,\#Z_0\myspl\star^n\mydollar) \myid{} (q_2,\varepsilon,\#Z_0\myspl\star^n\mydollar).
	\end{align*}
	Conversely, take any $w \in L(A)$. By the construction of $A$, we can assume that $w = a^n b^m$ (where $n,m \in \mynat$) and its accepting calculation is in the following form for some $l \in \mynat$:
	\begin{align*}
		(q_0,a^n b^m, \#Z_0\myspl\!\mydollar) &\myid{}^l (q_0,a^{n-l} b^m, \#Z_0\,\star^l\myspl\!\mydollar) \\
		&\myid{}\phantom{^l} (q_1,a^{n-l} b^m, \#Z_0\,\star^l\myspl\!\mydollar) \myid{}^* (q_1, \varepsilon, \#Z_0 \myspl \star^l \mydollar) \myid{} (q_2,\varepsilon,\#Z_0\myspl\star^l\mydollar).
	\end{align*}
	In this calculation, the steps $\myid{}^*$ indicate $n = l = m$.
\end{proof}
 \section{Calculation for the rewb of Theorem~\ref{thm:antisl}}
\label{app:calc}
Letting $k=4$, we calculate the language $L(\alpha)$ described by the $k$-rewb
\[
	\alpha = ((_1 \bs 4\,a )_1\,(_2 \bs 3 )_2\,(_3 \bs 2\,a)_3\,(_4 \bs 1 \bs 3 )_4)^\ast.
\]
			It is easy to see that
	\[	
		\reflang_k(\alpha) = \syuugou[ (\lbrack_1 4\,a\rbrack_1\,\lbrack_2 3 \rbrack_2\, \lbrack_3 2\,a\rbrack_3\,\lbrack_41\,3\rbrack_4)^n ]{n \in \mynat}.
	\]
		Let $b_n$ denote $(\lbrack_1 4\,a\rbrack_1\,\lbrack_2 3 \rbrack_2\, \lbrack_3 2\,a\rbrack_3\,\lbrack_41\,3\rbrack_4)^n$ (therefore $L(\alpha) = \syuugou[\deref_k(b_n)]{n \in \mynat}$). We show by induction that $\deref_k(b_n) = \deref_k(c_1 \cdots c_n)$ where $c_n \triangleq [_1 a^{n(n+1)/2} ]_1\,[_2 a^{n-1} ]_2\,[_3 a^n ]_3\,[_4 a^{n(n+3)/2}]_4$ for every $n \geq 1$. First, when $n = 1$,
		\begin{align*}
			\deref_k(b_1) &= \deref_k(\lbrack_1 4\,a\rbrack_1\,\lbrack_2 3 \rbrack_2\, \lbrack_3 2\,a\rbrack_3\,\lbrack_41\,3\rbrack_4) = \deref_k(\lbrack_1 a\rbrack_1\,\lbrack_2 3 \rbrack_2\, \lbrack_3 2\,a\rbrack_3\,\lbrack_41\,3\rbrack_4) \\
			&= \deref_k(\lbrack_1 a\rbrack_1\, \underline{\lbrack_2 \rbrack_2}\, \lbrack_3 2\,a\rbrack_3\,\lbrack_41\,3\rbrack_4) = \deref_k(\underline{\lbrack_1 a\rbrack_1}\,\lbrack_2 \rbrack_2\, \lbrack_3 a\rbrack_3\,\lbrack_41\,3\rbrack_4) \\
			&= \deref_k(\lbrack_1 a\rbrack_1\,\lbrack_2 \rbrack_2\, \underline{\lbrack_3 a\rbrack_3}\,\lbrack_4 a \,3\rbrack_4) = \deref_k(\lbrack_1 a\rbrack_1\,\lbrack_2 \rbrack_2\, \lbrack_3 a\rbrack_3\,\lbrack_4 a \,a\rbrack_4) = \deref_k(c_1).
		\end{align*}
		Next, in the case of $n+1$,
		\begin{align*}
					\deref_k(b_{n+1}) &= \deref_k(b_n\,\lbrack_1 4\,a\rbrack_1\,\lbrack_2 3 \rbrack_2\, \lbrack_3 2\,a\rbrack_3\,\lbrack_41\,3\rbrack_4) \\
					& = \deref_k(c_1 \cdots c_{n}\lbrack_1 4\,a\rbrack_1\,\lbrack_2 3 \rbrack_2\, \lbrack_3 2\,a\rbrack_3\,\lbrack_41\,3\rbrack_4) \\
					&= \deref_k(c_1\cdots c_n \lbrack_1 a^{n(n+3)/2}\,a\rbrack_1\,\lbrack_2 a^n \rbrack_2\, \lbrack_3 a^n\,a\rbrack_3\,\lbrack_4a^{n(n+3)/2}\,a\,a^n\,a\rbrack_4) \\
					&= \deref_k(c_1\cdots c_n \lbrack_1 a^{(n+1)(n+2)/2}\rbrack_1\,\lbrack_2 a^n \rbrack_2\, \lbrack_3 a^{n+1}\rbrack_3\,\lbrack_4a^{(n+1)(n+4)/2}\rbrack_4)\\
					&= \deref_k(c_1 \cdots c_n c_{n+1}).
		\end{align*}
	\[
		\therefore \zettaiti{\deref_k(b_n)} = \sum_{i=1}^{n}{\zettaiti{g(c_i)}} = \sum_{i=1}^{n}{(i^2 + 4i-1)} = \dfrac{n(n+7)(2n+1)}{6}.\quad (\text{also for }n=0)
	\]
	Therefore, we have $L(\alpha) = \syuugou[a^{n(n+7)(2n+1)/6}]{n \in \mynat}$.
 
\end{document}